\documentclass[12pt]{article}
\usepackage{amsmath, amssymb, amscd, amsthm, amsfonts}
\usepackage{graphicx}
\usepackage{hyperref}

\usepackage[toc,page]{appendix}
\usepackage{subcaption}
\usepackage{caption}
\usepackage{float}
\usepackage{algorithm}
\usepackage{algpseudocode}

\newcommand{\keywords}[1]{\noindent\textbf{Keywords:} #1}

\newtheorem{remark}{Remark}
\newtheorem{theorem}{Theorem}
\newtheorem{definition}{Definition}
\newtheorem{lemma}{Lemma}

\newtheorem{corollary}{Corollary} 

\date{\vspace{-4ex}}

\title{Linking Path-Dependent and Stochastic Volatility Models}

\author{Samuel N. Cohen\thanks{Mathematical Institute, University of Oxford, cohens@maths.ox.ac.uk} \and Cephas Svosve\thanks{Mathematical Institute, University of Oxford, cephas.svosve@maths.ox.ac.uk}}

\usepackage{biblatex}
\begin{document}
\maketitle
\begin{abstract}
We explore a link between stochastic volatility (SV) and path-dependent volatility (PDV) models. Using assumed density filtering, we map a given SV model into a corresponding PDV representation. The resulting specification is lightweight, improves in-sample fit, and delivers robust out-of-sample forecasts. We also introduce a calibration procedure for both SV and PDV models that produces standard errors for parameter estimates and supports joint calibration of SPX/VIX smile.
\end{abstract}


\section{Introduction}
In this study, we demonstrate that there exists an explicit connection between path-dependent volatility models and stochastic volatility models.

A key challenge in finance is estimating an asset's risk. In practice, this risk is often quantified by the instantaneous volatility of returns. However, since volatility is a latent variable, it cannot be directly computed from discrete observations of price data and therefore must be estimated. This issue is further complicated by the fact that volatility is neither constant nor stationary but evolves over time -- an assumption widely accepted in the literature  (see for example Pagan and Schwert \cite{Pagan1990ALTERNATIVEVOLATILITY}). Therefore, when estimating instantaneous volatility, we are attempting to derive an accurate value for a dynamic volatility based on limited data.

Two approaches that have emerged to describe instantaneous volatility are stochastic volatility (SV) and path-dependent volatility (PDV) models. SV models assume that volatility follows its own stochastic process, evolving according to a Stochastic Differential Equation (SDE). This process has its own innovations (i.e. source of randomness), which can be independent of or correlated with the innovations of the returns process. Early contributions include Vasicek's~\cite{Vasicek1977ANSTRUCTURE}  and Hull and White's~\cite{HULL1987TheVolatilities}  stochastic models for interest rate modeling and Heston's~\cite{Heston1993TheOptions} stochastic volatility model for stock option pricing, along with further advancements by Hagan, Kumar, Lesniewsky, and Woodward \cite{Hagan2002ManagingRisk} among others.

On the other hand, PDV models assume that volatility is primarily path-dependent, implying that past prices or returns explain most of the variation when predicting future realized or implied volatility. Unlike traditional stochastic volatility models, PDV models do not introduce an independent source of randomness for volatility. Instead, volatility fluctuations are assumed to arise from the accumulation of historical return variations. As a result, instantaneous volatility is driven by the same stochastic factors as the price process, leading to a complete market. This area of research traces back to the foundational work of Hobson and Rogers \cite{Hobson1998CompleteVolatility}, advancing through Zumbach \cite{Zumbach2005VolatilityTrends} and Chicheportiche and Bouchaud \cite{Chicheportiche2012TheSelf}, leading up to Guyon and Lekeufack's~\cite{Guyon2023VOLATILITYGUYON} and Nutz and Valdevenito's~\cite{Nutz2023OnModel} more recent contributions.

Guyon and Lekeufack \cite{Guyon2023VOLATILITYGUYON} demonstrate that PDV models generally outperform traditional approaches, both in their ability to model implied and realized volatilities and in their capacity to jointly calibrate SPX and VIX smiles. Further comparisons between PDV and SV models can be found in Foschi and Pascucci \cite{Foschi2008PathVolatility}. The prevailing narrative is that some PDV models are superior to SV models due to their high $R^2$ values when predicting volatility. 

However, while SV models and PDV models are often presented as distinct and separate, with some authors directly comparing the outperformance of one type over another, we explore how these models are explicitly related as one can begin with a specific SV model and derive its corresponding PDV representation. The accuracy of the PDV prediction is then linked to the inherent randomness of a corresponding SV model. For instance, we can start with a stochastic volatility model of the type presented by Heston~\cite{Heston1993TheOptions} and then derive a variation on the Guyon and Lekeufack\cite{Guyon2023VOLATILITYGUYON} PDV model as its corresponding representation.

That there is a relationship between PDV and SV models is well known based on the results of Brunick and Shreve \cite{Brunick2013MimickingEquation}. These  predate PDV models and show that for any general Itô process \(dS_t = \sigma_t S_t dW_t\), there exists a ``Markov mimicking" process of the form \(d\hat{S}_t = \sigma(t, (\hat{S}_{t_i})_{{t_i} \leq t}) \hat{S}_t dW_t\), such that \((S_t)_{t \geq 0}\) and \((\hat{S}_t)_{t \geq 0}\) have the same probability law. In particular, the statistical properties of a price process \((S_t)_{t \geq 0}\) produced by any SV or stochastic local volatility (SLV) model can be exactly reproduced by some PDV model. This is discussed further in the work of Carmona and Lacker \cite{Carmona2024MimickingKilling}. While this result implies evidence of a relationship between general path-dependent volatility, \(\sigma(t, (S_{t_i})_{{t_i} \leq t})^2\), and stochastic volatility, \(\sigma_t^2\), it does not provide a concrete way of translating the dynamics between path-dependent and stochastic volatility models.

Our work applies filtering theory to formalize the connection between PDV and SV models with explicit volatility dynamics, such as those of Heston \cite{Heston1993TheOptions} and Vasicek \cite{Vasicek1977ANSTRUCTURE}. As discussed in~\cite{Sarkka2013Bayesian_Filtering_and_Smoothing,Chan2021LargeVolatility}, filtering theory involves estimating a latent state variable based on noisy observations from available measurements. It has been successfully applied in various domains, including landmark applications, such as location estimation during the Apollo moon mission~\cite{Hun2016KalmanApplications}, and advanced navigation systems for self-driving cars~\cite{Fox2003BayesianEstimation,Chen2012KalmanSurvey,Welch1995AnFilter} and in finance as seen in the works of Javaheri, Lautier, and Galli \cite{javaheri2003filtering} and Bedendo and Hodges \cite{Bedendo2005TheApproach}. Instantaneous volatility is an example of a latent (hidden) variable that we cannot directly compute from discrete observations (cf. Andersen, Bollerslev, Diebold, and Ebens \cite{AndersenTorbenG2008TheVolatility} and Barndorff-Nielsen and Shephard \cite{barndorff2002econometric}).  When estimating volatility, one aims to determine both historical and current levels of a continuous-time volatility process. However, we can only observe realized prices or returns in discrete-time. Since volatility affects returns, these returns can be viewed -- in the context of filtering theory --  as noisy sensors for the unobservable stochastic volatility. This setup leads to a filtering problem which, in its discrete form, is intrinsically path-dependent.
Variants of this relationship have been considered, among others, by  Frey and Runggaldier \cite{frey2001}, Cvitani\'c, Lipster, and Rozovskii \cite{cvitanic2006}, and Aihara, Bagchi, and Saha \cite{AIHARA20086490}.  

A key challenge in the application of filtering theory to stochastic volatility is that the equations involved are typically infinite dimensional, and each of the above papers addresses this issue using a different method. For this reason, we establish this link using an Assumed Density Filter (ADF) approximation, which allows a lightweight approximation method with a clear interpretation.  The ADF (see~by~ Kushner \cite{Kushner1967ApproximationsFilters}) is a particular type of approximate Bayesian filtering technique. Unlike the Kalman filter, which considers only Gaussian distributions, an ADF accommodates a broader range of probability distributions by assuming a parametric form for the state density. ADFs are closely related to the projection filters studied by Brigo, Hanzon, and Le Gland \cite{Bao2024ApplicationEquations}. In our context, the key advantage of the ADF is that it naturally leads to a PDV interpretation of the filtered model.

Understanding this connection between PDV and SV models allows alternative calibration methods for both models. This approach not only provides adjusted PDV models with  fast, lightweight calibrations and improved in-sample fit but also extends effectively to better out-of-sample \(R^2\)-values. Furthermore, the relationship between PDV and SV models enables us to leverage PDV models for joint calibration of SPX and VIX option smiles while benefiting from the superior capabilities of SV models in Monte Carlo simulations for derivative pricing and scenario analyses. The findings of this study position PDV models within the broader context of SV models, emphasizing their potential synergies.

The remainder of the paper is organized as follows: Section~2 reviews common volatility models and discusses numerical discretization of continuous volatility. Section~3 introduces the mapping from SV to PDV using filters, with particular emphasis on the Assumed Density Filter (ADF). Section~4 illustrates the approach through a Heston volatility example in the framework of Guyon and Lekeufack \cite{Guyon2023VOLATILITYGUYON} with Section~4.1 comparing ADF and PDV models using simulated volatility, while Section~4.2 extends the comparison to real data, fitting to S\&P~500 realized volatility and the VIX. Section~5 concludes, and additional results are provided in the Appendix.
\section{Common volatility models}
In this section, we review common stochastic and path-dependent volatility models.

\subsection{Stochastic volatility models}
\begin{definition}
\label{def:generalsde}
 By a continuous-time stochastic volatility model, we refer to a pair of real-valued processes \((X,\nu)\) governed by the stochastic differential equations
    \begin{equation}
        \label{eqn:svsystem}
        \begin{split}
            dX_t &= u(\nu_t, X_t) \, dt + g(\nu_t, X_t) \, dW_t^i, \\
            d\nu_t &= \varphi( \nu_t) \, dt + \zeta(\nu_t) \, dW_t^j,
        \end{split}
    \end{equation}
where \(u,g, \varphi\) and \(\zeta\) are sufficiently regular functions, and \( W^i \) and \( W^j \) are (scalar) Brownian motions correlated by \(d\langle W^i, W^j\rangle_t = \rho \, dt\) where \(\rho\in [-1,1]\) is a constant.
\end{definition}
Here, \( \nu \) is interpreted as a continuous ``volatility" process, which determines the instantaneous volatility of the log price process (or forward rate in a fixed income context) \(X\). 

\begin{remark}
A key point is that the volatility process introduces an additional source of randomness \( W^j \), which is distinct from the randomness \( W^i \) already embedded in the process \( X\).
\end{remark}

Various SV models in the literature cater to specific asset classes and scenarios. For instance, Hagan et al.~\cite{Hagan2002ManagingRisk} developed the Stochastic Alpha, Beta, Rho (SABR) model to describe volatility dynamics in forward rate modeling.
\begin{equation}\begin{split}
    d{F_t} &= {F_t^{\beta}}{\sigma_t}dW_t^{i},\\
    d{\sigma_t} &= \xi {\sigma_t}dW_t^{j},
\end{split}\label{eq:SABRmodel}
\end{equation}
where \(F\) is the forward rate $\xi> 0$ is the volatility of volatility, and $\beta\in [0,1]$ is a constant. 

Similarly, the Heston model \cite{Heston1993TheOptions} specifies the instantaneous volatility as a Cox--Ingersoll--Ross process \cite{Cox1985ARATES1}. The SDE  is given by,
\begin{equation}
\label{eqn:Heston1.4}
\begin{split}
  dX_t &= (\tilde{\mu}-\tfrac{1}{2}\nu_t)\cdot dt + \sqrt{\nu_t}\cdot dW_t^{i},\\
d\nu_t &= \kappa(\theta -\nu_t)\cdot dt + \xi\sqrt{\nu_t}\cdot dW_t^{j},
\end{split}
\end{equation}

where $\kappa$ is the rate at which $\nu$ reverts to the long-run average volatility $\theta$, while $\xi$ represents the volatility of volatility. $X$ is the log-price process and $\tilde{\mu}$ is the corresponding long-term mean of simple returns.  The Feller condition, \(2\kappa\theta \geq \xi\), ensures that the volatility remains non-negative while maintaining mean reversion (see Albrecher, Mayer, Schoutens, and Tistaert \cite{albrecher2007little}). Again, the model has two sources of randomness; $W^{i}$ and $W^{j}$ correlated by \(\rho \in [-1,1]\).

Another stochastic volatility model with a formulation close to the Heston model is the continuous-time version of the GARCH(1,1) model proposed by Drost and Werker \cite{drost1996closing},
\begin{equation}
\label{eqn:2}
\begin{split}
dX_t &= (\tilde{\mu}-\tfrac{1}{2}\nu_t)\cdot dt + \sqrt{\nu_t}\cdot dW_t^{i},\\
d\nu_t &= \kappa(\theta -\nu_t)\cdot dt + \xi{\nu_t}\cdot dW_t^{j},
    \end{split}
\end{equation}
where all parameters have definitions similar to those in the Heston model in \eqref{eqn:Heston1.4}. Further examples include the Vasicek model~\cite{Vasicek1977ANSTRUCTURE} and additional variations due to Hull and White \cite{HULL1987TheVolatilities}.

\subsection{Path-dependent volatility models}
\begin{definition}
    
\label{def:generalpde}
 By a continuous-time path-dependent volatility model, we refer to a pair of real-valued processes \((X,\nu)\) governed by the system:
\begin{equation}
\label{eqn:pdvsystem}
    \begin{split}
            dX_t &= u(\nu_t, X_t) \, dt + g(\nu_t, X_t) \, dW_t, \\
            d\nu_t &= \varphi( \{X_{s}\}_{s< t})dt. 
    \end{split}
\end{equation}

\end{definition}
\begin{remark}
  In a path-dependent model, the volatility process does not introduce any new randomness beyond that inherent in the price process. Instead, volatility depends on an asset's historical prices, which results in a price process that is not Markovian. We can also see PDV models as generalizing the local volatility models considered by Dupire et al. \cite{Dupire1994PricingSmile}, where \(\nu\) is simply a function of \(X\).
\end{remark}

\begin{remark}
Guyon and Lekeufack \cite[Section 2.4]{Guyon2023VOLATILITYGUYON} give a further classification of models in-between path-dependent and stochastic volatility cases (as we have defined them), depending principally on the level of correlation $\rho$, and whether the stochastic volatility dynamics are permitted to depend on the path of the price process (which we have excluded above). What we shall see in discrete time (where our observations are typically obtained), is that this is not the primary distinction between these models; rather, a stochastic volatility model can be seen (in the filtration generated by discrete price observations) as a noisy path-dependent model, even if the stochastic volatility process has no dependence on the price process in its original formulation.
\end{remark}

Zumbach \cite{Zumbach2005VolatilityTrends} identified a stylized fact, now known as the Zumbach effect, that past returns forecast future volatility more effectively than past volatility forecasts future returns. This suggests that path-dependent models may capture the principal effects of volatility in financial markets. These findings align with Hobson and Rogers \cite{Hobson1998CompleteVolatility}, who demonstrate that past log-prices effectively predict future realized volatility. More precisely, letting \( Z_t = X_t-rt \), where \( X_t \) is the log-price at time \( t \), and \( r \) is the risk-free interest rate, given a decay rate \(\lambda>0\), the model introduced by Hobson and Rogers \cite{Hobson1998CompleteVolatility} is 
\begin{align*}
dZ_t &= \mu(S_t^{(1)}, \ldots, S_t^{(n)}) dt + \sigma(S_t^{(1)}, \ldots, S_t^{(n)}) dB_t ,\\
S_t^m &= \int_0^\infty \lambda e^{-\lambda u} \big(Z_t - Z_{t-u}\big)^m du.
\end{align*}
In particular, the mean \(\mu(S_t^{(1)}, \ldots, S_t^{(n)})\) and volatility \(\sigma(S_t^{(1)}, \ldots, S_t^{(n)})\) are path-dependent functions, as they depend on historical prices with exponential weightings. 

Foschi and Pascucci \cite{Foschi2008PathVolatility} extend this work to a more general formulation. They first define a strictly positive average weights function \(\varphi\) and a normalizing factor,
\(
\Phi(t) = \int_{-\infty}^t \varphi(s) \, ds.
\)
A typical choice of the average weights function is \(\varphi(s) = e^{-\lambda (t-s)} \text{ where } \lambda \geq 0\) for some constant future time \(t\). The average process is defined by
\[
M_t = \frac{1}{\Phi(t)} \int_{-\infty}^t \varphi(s) Z_s \, ds, \quad t \in (0, T],
\]
or equivalently,
\[
dM_t = \frac{\varphi(t)}{\Phi(t)} \big(Z_t - M_t\big) \, dt.
\]
Here, \( Z_t = X_t - rt \) is the log-discounted price process whose dynamics are,
\[
dZ_t = \mu(Z_t - M_t) \, dt + \sigma(Z_t - M_t) dW_t,
\]
where \(\mu\) and \(\sigma\) are bounded, Hölder continuous functions, and \(\sigma\) is uniformly strictly positive. The volatility process \(\sigma(Z_t - M_t)\) depends on past prices through \(M_t\), making it a function of historical prices. Guyon and Lekeufack \cite{Guyon2023VOLATILITYGUYON}'s continuous-time model belongs to the same PDV category, with slight modifications. It models volatility as a function of past returns, thereby aligning with the Zumbach effect discussed earlier.

\subsection{Discrete-time volatility models}
In continuous time, volatility can be inferred from the price process \( X \), as it corresponds to the derivative of the quadratic variation -- assuming \( X \) satisfies a stochastic differential equation such as \eqref{eqn:svsystem}. However, in practice, financial data is observed at discrete time intervals, with returns sampled either regularly or irregularly. As a result, any estimate of instantaneous volatility is inherently path-dependent. When prices are observed only in discrete time, it is generally not possible to obtain a perfect estimate of current or future volatility. This limitation introduces prediction error in the forecasting of next-day volatility.

There are several common discrete-time models for volatility, such as the Autoregressive Conditional Heteroscedasticity (ARCH) model introduced by Engle  \cite{EngleRobertF.1982AutoregressiveInflation}, the Generalized Autoregressive Conditional Heteroscedasticity (GARCH) model proposed by Bollerslev \cite{Bollerslev1987}, and the Exponential GARCH (EGARCH) model by Nelson \cite{NelsonDanielB.1991ConditionalReturns}. 
A GARCH \((p, q)\) model -- where \(p\) represents the order of the GARCH terms \( (\sigma^2)\) and \(q\) is the order of the ARCH terms \((\epsilon^2)\) -- is a discrete-time model of the form,
\begin{align*}
X_t &= X_{t-1} + \mu \, + \,\epsilon_{t}\\
\sigma_t^2 &= \omega + \sum_{i=1}^{q} \alpha_i\,\epsilon_{t-i}^2 + \sum_{i=1}^{p} \beta_i \,\sigma_{t-i}^2, \text{ where}\quad \epsilon_t  \sim \mathcal{N}(0, \sigma_t^2),
\end{align*}
As before, \( X_t \) represents the log-price at time \( t \), while \( \mu \) denotes the long-term mean of log-returns. The parameter \( \omega \) is a constant, while the coefficients \( \alpha_i \) correspond to past squared error terms \( \epsilon_{t-i}^2 \), forming the \( q \) component in the GARCH(\( p, q \)) model. The coefficients \( \beta_i \) are associated with past volatility \( \sigma_{t-i}^2 \), capturing its persistence over time and representing the \( p \) component in the GARCH(\( p, q \)) model.  

We emphasize that this model introduces no additional source of randomness, as the \( \epsilon \) terms are rather representing the innovations already embedded in the log-price process \( X \). In this sense, this model falls in the discrete-time PDV class.

\subsubsection*{Guyon and Lekeufack's discrete-time PDV model}
Guyon and Lekeufack \cite{Guyon2023VOLATILITYGUYON} propose a discrete-time model which appears to perform well when jointly calibrating the SPX/VIX smile. Suppose \(\sigma_{\nu}(t)\) represents the implied volatility observed at time \(t\), and \(\sigma_{x}(t)\) denotes the integrated realized volatility measured over  \([t, t+1)\) (i.e., the volatility realized over day \(t+1\)). The relationship between volatility and past returns is modelled by,
\[
\begin{pmatrix} \sigma_{\nu}(t) \\ \sigma_{x}(t) \end{pmatrix}
= \begin{pmatrix} \beta_0 & \beta_1 & \beta_2 \\ \alpha_0 & \alpha_1 & \alpha_2 \end{pmatrix}
\begin{pmatrix}
1 \\ R_{1,t} \\ \Sigma_t
\end{pmatrix},
\]
where \(R_{1,t} = \sum\limits_{t_i \leq t} K^\lambda(t-t_i)\,Y_{t_i},\) and \( \Sigma_t  = \sqrt{\sum\limits_{t_i \leq t} K^\lambda(t-t_i)Y_{t_i}^2}\). 

\noindent The (simple) returns are given by \(Y_{t_i} = \frac{S_{t_i} - S_{t_{i-1}}}{S_{t_{i-1}}}\), with conditional variance \(\operatorname{Var}\Big(Y_{t_i}\Big|\{Y_{t_j}\}_{t_j<t_i}\Big)\) where \(S_{t_i}\) represents the price at time \(t_i\). The weighting factor 
\(K^\lambda(\tau) = \lambda e^{-\lambda\tau},\)
depends on a decay rate \(\lambda\). The other parameters are such that, \[
\alpha_0, \beta_0 \geq 0, \quad \alpha_1, \beta_1 \leq 0, \quad \alpha_2, \beta_2 \in (0,1).
\]

The \( R^2 \) values for Guyon and Lekeufack's \cite{Guyon2023VOLATILITYGUYON}'s PDV model are remarkably high -- reaching up to \(0.91\) for predicting the current day's implied volatility (as measured by the VIX index) while using only S\&P 500 returns data without incorporating any of the historical VIX data. However the \(R^2\) values decrease significantly when applied to forecasting the next day's realized volatility in place of implied volatility. Specifically, Guyon and Lekeufack's~\cite{Guyon2023VOLATILITYGUYON} out-of-sample \( R^2 \) values for next-day realized volatility are listed in Table \ref{table:R^2}.

\begin{table}[h]
\centering
\begin{tabular}{|c|c|c|c|c|c|}
\hline
\textbf{Index} & {SPX} & {STOXX} & {FTSE 100} & {DAX} & {NIKKEI} \\
\hline
\textbf{\(R^2\)} & 0.654 & 0.682 & 0.617 & 0.557 & 0.504 \\
\hline
\end{tabular}
\caption{\( R^2 \) values for the prediction of next-day's realized volatility.}
\label{table:R^2}
\end{table}
\noindent The corresponding out-of-sample \( R^2 \) values for current-day implied volatility are in Table \ref{table:R^2_2}. 
\begin{table}[h]
\centering
\begin{tabular}{|c|c|c|c|c|c|}
\hline
\textbf{Index} & {VIX} & {VSTOXX} & {IVI} & {VDAX-NEW} & {Nikkei 225 VI} \\
\hline
\textbf{\(R^2\)} & 0.855 & 0.913 & 0.870 & 0.918 & 0.800 \\
\hline
\end{tabular}
\caption{\( R^2 \) values for the prediction of current-day's implied volatility.}
\label{table:R^2_2}
\end{table}

While the results remain consistent with the conclusion that past returns account for much of the variation in volatility, they also reveal a significant degree of unexplained variation (especially in the forecasts of next-day realized volatilities) -- suggesting omitted sources of noise in the model. All at the same time, the model's performance remains impressive, so we seek to understand how it can be justified within a stochastic volatility context.

 \subsection{Discrete-time approximations of stochastic volatility}
Given a continuous-time stochastic volatility process, an approximate discrete solution can be obtained using discretization methods, such as the Euler--Maruyama or Milstein methods (see, for example, Glasserman \cite{Glasserman2004MonteEngineering}).

\begin{lemma}
\label{corollary:eulermaruyama_approx}
Suppose we have a stochastic volatility process defined in \eqref{eqn:svsystem}. By applying the Euler–Maruyama approximation, we derive the approximate discrete dynamics
\begin{equation}
\label{eqn:discretedynamics}
\begin{split}
 X_{t_n} &= X_{t_{n-1}} + u( \nu_{t_{n-1}},X_{t_{n-1}}) h\, + g(\nu_{t_{n-1}},X_{t_{n-1}}) \, \sqrt{h}Z_n^{i},\\
    \nu_{t_n}  &= \nu_{t_{n-1}}  + h\varphi(\nu_{t_{n-1}}) + \zeta( \nu_{t_{n-1}})\sqrt{h}Z_n^{j},
\end{split}
\end{equation}
where \( h = t_n - t_{n-1} \) is a fixed time step size, and \( Z_n^{i}, Z_n^{j} \sim \mathcal{N}(0, 1) \) are white noise processes. 
\end{lemma}

\begin{remark}
    We use the \(t_{n-1}\) variance, \(g(\nu_{t_{n-1}},X_{t_{n-1}})\), to simulate returns over the interval \([t_{n-1}, t_n)\). This ensures that our scheme is explicit i.e. the volatility associated with the returns generated  over \([t_{n-1}, t_n)\) is known at the start of the time step.
\end{remark}
\begin{remark}
The Euler--Maruyama scheme converges to the continuous-time solution in both the weak and strong senses under appropriate conditions. For instance, when the drift and diffusion coefficients satisfy global Lipschitz and linear growth conditions, the approximation achieves a weak convergence rate of order \(1\) and a strong convergence rate of order \(\frac{1}{2}\) (see Glasserman \cite{Glasserman2004MonteEngineering} or Kloeden and Platen \cite{KloedenPlaten1989}).
\end{remark}

The Euler--Maruyama scheme does not guarantee that the volatility process remains strictly positive. In this study, we apply a positive truncation to our numerics, setting any negative values to zero. The probability of the process hitting zero or becoming negative is generally low, especially in the Heston model when the Feller condition is satisfied.

\section{From SV to PDV using filters}
In this section, we demonstrate how a stochastic volatility model can be explicitly represented as a modified PDV model. Our key assumption is that observations occur in discrete-time. 

From a Bayesian filtering perspective, we begin with a prior belief about the volatility distribution and then observe the returns that depend on it. By considering our on-the-fly estimates of volatility given the return observations, we obtain a model similar to many PDV models. However, this typically leads to an infinite-dimensional prior density problem. To address this complexity, our analysis employs the Assumed Density Filter, which approximates the infinite-dimensional density with a finite-dimensional family.

\subsection{Discrete-time Bayes filtering}
We begin by revisiting discrete-time filtering, which consists of two key stages: correction and prediction. We consider a filtration \( \mathcal{Y}_{t_n} = \{Y_0, Y_1, \dots, Y_{t_n}\} \), representing the discrete-time information from past observations, where \( Y_{t_n} \in \mathbb{R} \).

\begin{definition}
\label{def:filtering}
Let \( \nu_{t_n} \in \mathbb{R} \) denote the hidden state variable at time \( t_n \), and let \( \mathcal{Y}_{t_{n-1}} \) represent the collection of discrete-time observable variables up to time \( t_{n-1} \). We denote by \( p(\cdot \mid \mathcal{Y}_{t_{n-1}}) \), the conditional law of \( \nu \) given \( \mathcal{Y}_{t_{n-1}} \). Filtering recursively applies two steps:

\begin{enumerate}
 \item \textbf{Prediction:} Before observing \( Y_{t_n} \), the prediction of \( \nu_{t_n} \) based on past observations is given by its conditional prior density, obtained through the Chapman--Kolmogorov equation
    \begin{equation}
    \label{eqn:bayesprior}
    p(\nu_{t_n} \mid \mathcal{Y}_{t_{n-1}}) = \int_{\mathbb{R}} p(\nu_{t_n} \mid \nu_{t_{n-1}}) \, p(\nu_{t_{n-1}} \mid \mathcal{Y}_{t_{n-1}}) \, d\nu_{t_{n-1}}.
    \end{equation}
    
    \item \textbf{Correction:} Conditioning on both the past observations and the new recent observation, \( Y_{t_n} \), the correction step updates the prior to a posterior density of \( \nu_{t_n} \). As  \( Y_{t_n} \) is conditionally independent of \(\mathcal{Y}_{t_{n-1}}\) given \( \nu_{t_n} \), this is done using Bayes' theorem:
    \begin{equation}
    \label{eqn:bayespriorcorrection}
    \begin{split}
    p(\nu_{t_n} \mid \mathcal{Y}_{t_n}) &= \frac{ p(Y_{t_n} \mid \nu_{t_n}) \, p(\nu_{t_n} \mid \mathcal{Y}_{t_{n-1}})}{\int_{\mathbb{R}} p(Y_{t_n} \mid \nu_{t_n}) \, p(\nu_{t_n} \mid \mathcal{Y}_{t_{n-1}}) \, d\nu_{t_n}} \\[8pt]
    &\propto p(Y_{t_n} \mid \nu_{t_n}) \, p(\nu_{t_n} \mid \mathcal{Y}_{t_{n-1}}).
    \end{split}
    \end{equation}
\end{enumerate}

The filtering process iteratively applies the prediction and correction steps to generate a sequential estimate of the hidden state over time.
\end{definition}

The difficulty is that the density \(p(\cdot \mid \mathcal{Y}_{t_{n}}) \) is an infinite dimensional object. Ionides and Wheeler \cite{Ionides2024ModelFiltering} consider using a particle filter approximation to estimate current volatility. We will instead use the ADF approximation to give a low-dimensional approximate filter, which can be analysed explicitly, and where we can draw comparisons with PDV models.

\subsection{Assumed Density Filtering (ADF)}
Bayesian filtering is well-suited for scenarios where the density is known exactly. However, in many nonlinear filtering problems, the density in \eqref{eqn:bayesprior} may be unknown or infinite-dimensional and the normalization in \eqref{eqn:bayespriorcorrection} may be difficult, making the filter computationally intractable. In such cases, an Assumed Density Filter offers a practical alternative. This approach approximates \eqref{eqn:bayesprior} and \eqref{eqn:bayespriorcorrection} with a density that has finite dimensional statistics while closely resembling the true density.  

A key advantage of this technique is that, while the true density may be infinite-dimensional -- requiring additional parameters for each new time point -- a single finite dimensional approximate density can serve as a good approximation collectively across all time points.

\begin{definition}
\label{definition:adf}
We define an ADF as a Bayesian filter in which the prediction and/or correction steps are approximated using a specified parametric form for the density. In particular, the prediction step is governed by a prior density, 
\begin{align}
\label{eqn:approximate_symbol}
\pi(\nu_{t_n}\mid \mathcal{Y}_{t_{n-1}}) \approx p(\nu_{t_n} \mid \mathcal{Y}_{t_{n-1}}).
\end{align}
The corresponding correction step is a posterior density,
\begin{equation}
\begin{split}
\pi(\nu_{t_n} \mid \mathcal{Y}_{t_n}) &\approx \frac{p(Y_{t_n} \mid \nu_{t_n})\, \pi( \nu_{t_n}\mid \mathcal{Y}_{t_{n-1}})}{\int_{\mathbb{R}} p(Y_{t_n} \mid \nu_{t_n}) \pi(\nu_{t_n}\mid \mathcal{Y}_{t_{n-1}}) \, d\nu_{t_n}} \\[8pt]
&\propto p(Y_{t_n} \mid \nu_{t_n}) \pi(\nu_{t_n}\mid\mathcal{Y}_{t_{n-1}}),
\end{split}
\end{equation}
where, \( \nu_{t_n}, Y_{t_n} \in \mathbb{R} \) represent the hidden state and observable variables, respectively, at time \( t_n \).

\end{definition}

Examples of approximate filters include the Extended Kalman Filter (EKF) \cite{Ding2022ExtendedIdentification}, which applies Kalman filtering techniques to non-linear systems by linearizing the state and measurement models around the current estimate and so obtain a Gaussian approximation. Other instances of approximate filtering include the particle filter \cite{Kunsch2013ParticleFilters} where the density is replaced by finitely many point-masses. In the continuous-time setting, Kushner \cite{Kushner1967ApproximationsFilters} laid the groundwork for ADFs by introducing the idea of approximating the true posterior distribution with a lower-dimensional density, see also the more general geometric filters of Brigo, Hanzon, and Le Gland \cite{Brigo1995}.

While there is no universal theory proving error bounds for ADFs, some studies have shown that the error in specific approximate filters, such as approximate Wonham filters, is uniformly bounded Cohen and Fausti \cite{Cohen2023ExponentialFilters}. In variational inference (see for example Dur\'an‑Mart\'in, Leandro, and Kevin \cite{Duran-Martin2024BONE:Environments}), ADFs are also often used, typically chosen to minimize the Kullback--Leibler (KL) error \cite{JeongBayesianFiltering, Jeong2019AssumedQ-learning, Ranganathan2004AssumedFiltering}.

\subsection{Volatility modelling by an ADF}
To highlight the connection between PDV and SV models, we will use a discrete-time Bayesian ADF, where the assumed density is approximated by a 2-dimensional family based on moment matching. This approximation is updated at each discrete time step to incorporate new information, thereby refining the assumed density's optimal parameters.

Suppose we are given particular stochastic volatility dynamics in continuous-time described by
\begin{equation}
\begin{aligned}
dX_t &= \mu\,dt + \sqrt{\nu_t}\,dW^{i}_t,\\
d\nu_t &= \varphi(\nu_t)\,dt + \zeta(\nu_t)\,dW^{j}_t,
\end{aligned}
\end{equation}
 where \(X_t = \log S_t\) and \(\mu\) is the mean of log returns. We leave the volatility process in general form to accommodate any volatility model within this class of dynamics.

\begin{remark}
By It\^o’s lemma, this implies the price process
\[
dS_t = \big(\mu + \tfrac{1}{2}\nu_t\big)S_t\,dt + \sqrt{\nu_t}S_t\,dW^{(S)}_t,
\]
which differs slightly from Heston's \cite{Heston1993TheOptions} original formulation, where the drift is the mean of simple returns. These two forms are equivalent, but the drift interpretation differs.
\end{remark}

Applying the discrete-time Euler--Maruyama approximation as in \eqref{eqn:discretedynamics} gives
\begin{equation}
\label{eqn:discretedynamicsb}
\begin{split}
 X_{t_n} &= X_{t_{n-1}} + \mu h\, + \sqrt{\nu_{t_{n-1}}\,h}Z_n^{i},\\
    \nu_{t_n}  &= \nu_{t_{n-1}}  + \varphi(\nu_{t_{n-1}})h + \zeta( \nu_{t_{n-1}})\sqrt{h}Z_n^{j},
\end{split}
\end{equation}

where \( X \) is a log-price process and \( \nu \) a volatility process. The price process in this system is driven by its own (Gaussian) noise process, \( Z_n^{i} \), while the volatility process is influenced by another exogenous noise process, \( Z_n^{j} \). For simplicity, we initially assume \( \mathbb{E}[Z_n^{i}, Z_n^{j}] = \rho = 0 \) and will later explore how to incorporate \( \rho \). Writing \(\mathcal{F}_{t_n} = \sigma(\{Z_m^{(i)},Z_m^{(j)}\}_{m\leq n})\) for the associated full-information filtration, the log-returns process can be expressed as  

\begin{equation}
\label{eqn:returns}
Y_{t_n}  = X_{t_n} - X_{t_{n-1}},
\end{equation}
and, under the Euler--Maruyama scheme, has a conditional Gaussian distribution,
\[ Y_{t_n}|\mathcal{F}_{t_{n-1}}  \sim \mathcal{N}\Big(\mu h, \nu_{t_{n-1}} h\Big),\]
We assume that we observe the returns process, while the instantaneous volatility \( \nu \)  remains unobservable.

Our goal is to estimate the unobservable instantaneous volatility using discrete-time Bayesian filtering. Due to the dynamic nature of volatility, this density generally may not have a finite dimensional parameterization so, we approximate the true density with a finite-moment assumed density. 

A natural choice for the assumed probability density is the Inverse Gamma distribution, which is well-suited for Bayesian filtering due to its convenience as a conjugate prior to the Gaussian likelihood commonly used for log-returns. Furthermore, it guarantees a positive volatility value and provides flexibility in modeling heavy-tailed behavior, making it effective for capturing volatility clustering -- a well-documented feature in financial markets as reported by Cont \cite{Cont2005VolatilityModels}.

\subsubsection{ADF correction}
We begin by illustrating the correction step, where we sequentially incorporate observed data as it becomes available.
\begin{lemma} 
\label{lemma:alphabetadef}  
Suppose the prior density \( \pi(\nu_{t_n} \mid \mathcal{Y}_{t_{n-1}}) \) is (approximated by) an Inverse Gamma density  
\[
\nu_{t_n} \mid \mathcal{Y}_{t_{n-1}} \sim \text{Inv-Gamma}(\alpha_{t_n| t_{n-1}}, \beta_{t_n|t_{n-1}}),
\]  
with shape parameter \( \alpha_{t_n|t_{n-1}} > 2 \) and scale parameter \( \beta_{t_n|t_{n-1}} > 0 \). The recursive correction (Definition \ref{def:filtering}) reduces to the parameter update  
\begin{equation}
\alpha_{t_n|t_n} = \alpha_{t_n|t_{n-1}} + \frac{1}{2}, \text{ and, }\, 
\beta_{t_n|t_n} = \beta_{t_n|t_{n-1}} + \frac{1}{2h} (Y_{t_n} - \mu h)^2.
\label{eqn:alphabetaupdate}
\end{equation} 
That is,
\[\nu_{t_n} \mid \mathcal{Y}_{t_{n}} \sim \text{Inv-Gamma}(\alpha_{t_n| t_{n}}, \beta_{t_n|t_{n}}).\]
This update is exact if the prior is Inverse Gamma and returns follow \eqref{eqn:returns}.
\end{lemma}

\begin{remark}
Before proving Lemma \ref{lemma:alphabetadef}, it is worth noting that choosing the shape parameter \( \alpha_{t_n|t_{n-1}} > 2 \) ensures the existence of both the expectation and the variance of the Inverse Gamma distribution.
\end{remark}

\begin{proof}[Proof of Lemma \ref{lemma:alphabetadef}]
The probability density function of an Inverse-Gamma distribution is 
\[
\pi_{(\nu_{t_n}~|~\mathcal{Y}_{t_{n-1}})}(x^2)\approx \frac{\left(\beta_{t_n|t_{n-1}}\right)^{\alpha_{t_n|t_{n-1}}}}{\Gamma(\alpha_{t_n|t_{n-1}})} \Big( \frac{1}{x^2} \Big)^{\alpha_{t_n|t_{n-1}} + 1} \exp\Big(-\frac{\beta_{t_n|t_{n-1}}}{x^2}\Big).
\]
The likelihood function \( f( Y_{t_n}|\nu_{t_{n}};\mu,h) \) is the conditional probability density function of Gaussian log-returns, where the returns \( Y_{t_n} \) follow \eqref{eqn:returns}. Then \eqref{eqn:bayespriorcorrection} becomes, up to a constant factor independent of \(\nu_{t_n}\),
\begin{align*}
\pi_{(\nu_{t_n}~|~\mathcal{Y}_{t_{n}})}(x^2) &\propto \pi_{(\nu_{t_n}~|~\mathcal{Y}_{t_{n-1}})}(x^2)  \times f( Y_{t_n}|\,x^2;\mu,h) \\
     &\propto \left[\left(\frac{1}{x^2}\right)^{\alpha_{t_n|t_{n-1}} +1}\exp\left({-\frac{\beta_{t_n|t_{n-1}}}{x^2}}\right) \right]\times \left[\left(\frac{1}{x^2h}\right)^{\frac{1}{2}}\exp\left({-\frac{(Y_{t_{n}}-\mu h)^2}{2x^2h}}\right)\right]\\
      &\propto \left(\frac{1}{x^2}\right)^{\left(\alpha_{t_n|t_{n-1}} +\frac{1}{2}\right)+1}\exp\left[{-\frac{1}{x^2}\left(\beta_{t_n|t_{n-1}}+\frac{1}{2h}\left( Y_{t_{n}}-\mu h\right)^2\right)}\right]. 
\end{align*}
    This is the density of a Normal Inverse-Gamma distribution, and if we define \(\alpha_{t_n|t_n}\) and \(\beta_{t_n|t_n}\) as in \eqref{eqn:alphabetaupdate}, we get the posterior
    \[
\pi_{(\nu_{t_n}~|~\mathcal{Y}_{t_{n}})}(x^2) \propto \left( \frac{1}{x^2} \right)^{\alpha_{t_n|t_n} + 1} \exp\left( - \frac{\beta_{t_n|t_n}}{x^2} \right).
    \]
   This is the density of an Inverse Gamma distribution with shape and scale parameters \(\alpha_{t_n|t_n}\) and \(\beta_{t_n|t_n}\) respectively. 
\end{proof}
\subsubsection{Predicting with the ADF}
We now approximate \eqref{eqn:bayesprior} in the form of an ADF. We do this by matching moments of the increments of the volatility dynamics in \eqref{eqn:discretedynamicsb} with the moments of the assumed density, the Inverse Gamma distribution. The expectation and variance of the Inverse Gamma distribution are given by  
\begin{equation}
\label{eqn:10a}
\begin{split}
\mathbb{E}\Big(\nu\Big)& =  \frac{\beta}{\alpha - 1}, \quad \text{and, }\\
 \operatorname{Var}
 \Big(\nu\Big)& = \frac{\beta^2}{(\alpha - 1)^2 (\alpha - 2)},
\end{split}
\end{equation}
respectively, where the shape and scale parameters \(\alpha\) and \(\beta\) are defined as in Lemma \ref{lemma:alphabetadef}. We match these with the moments of forward volatility dynamics derived in Lemma \ref{lemma:numericalmoments} below.

\begin{lemma}
\label{lemma:numericalmoments}
Given the volatility dynamics in \eqref{eqn:discretedynamicsb}, we can express the conditional moments of the discrete-time volatility \(\{\nu_{t_n}\}_{t_n \in \mathbb{N}}\) as
\begin{equation}
\label{eqn:11b}
\begin{split}
        \mathbb{E}\Big(\nu_{t_n} | \mathcal{Y}_{t_{n-1}}\Big) &= \mathbb{E}\Big(\nu_{t_{n-1}}|\mathcal{Y}_{t_{n-1}}\Big) + h \, \mathbb{E}\Big(\varphi( \nu_{t_{n-1}}) | \mathcal{Y}_{t_{n-1}}\Big), \\
    \operatorname{Var}\Big(\nu_{t_{n}}| \mathcal{Y}_{t_{n-1}}\Big) &= \operatorname{Var}\Big(\nu_{t_{n-1}} +\varphi(\nu_{t_{n-1}})h | \mathcal{Y}_{t_{n-1}}\Big) +  h \operatorname{Var}\Big(\zeta(\nu_{t_{n-1}})| \mathcal{Y}_{t_{n-1}}\Big).
\end{split}
\end{equation}
See Appendix \ref{appendix1:weak_convergence_euler_maruyama} for details of this calculation.
\end{lemma}
We construct the ADF through a stepwise matching of the conditional moments in \eqref{eqn:10a} and \eqref{eqn:11b} to get the result in Lemma \ref{lemma:moments} below.
\begin{lemma}  
\label{lemma:moments}  
The ADF from matching moments of the volatility process is given by the parameter update
\begin{equation}  
\label{eqn:shapeparametersasmoments}  
\begin{split}  
\alpha_{t_n|t_{n-1}} &=  \frac{\Big[\mathbb{E}\Big(\nu_{t_{n-1}} | \mathcal{Y}_{t_{n-1}}\Big) + h \mathbb{E}\Big(\varphi( \nu_{t_{n-1}})|\mathcal{Y}_{t_{n-1}}\Big)\Big]^2}{\operatorname{Var}\Big(\nu_{t_{n-1}} +\varphi(\nu_{t_{n-1}})h | \mathcal{Y}_{t_{n-1}}\Big) +  h \operatorname{Var}\Big(\zeta(\nu_{t_{n-1}})| \mathcal{Y}_{t_{n-1}}\Big)} + 2,\\[8pt]
\beta_{t_n|t_{n-1}}  &= \Big(\alpha_{t_n|t_{n-1}}-1\Big)\Big[\mathbb{E}\Big(\nu_{t_{n-1}} | \mathcal{Y}_{t_{n-1}}\Big) + h \mathbb{E}\Big(\varphi( \nu_{t_{n-1}})|\mathcal{Y}_{t_{n-1}}\Big)\Big].\notag  
\end{split} 
\end{equation} 
\end{lemma}  
\subsection{Connecting SV and PDV}
Having defined the Assumed Density Filter in terms of the forward dynamics of discrete volatility, we now illustrate the connection between stochastic and path-dependent volatility by applying back-propagation into the history of the discrete filtering process. This enables us to express the filter in terms of past returns.
\begin{theorem}
\label{theorem:generallinkbetweenSVandPDV1}
The ADF has the recursive statistics 
\begin{small}
\begin{equation}
\label{eqn:corrandpred}
\begin{aligned}
\mathbb{E}[\nu_{t_n} \mid \mathcal{Y}_{t_n}] &= \frac{(Q_{t_n|t_{n-1}} + 1)\big( \mathbb{E}[\nu_{t_{n-1}} \mid \mathcal{Y}_{t_{n-1}}] + h \, \mathbb{E}[\varphi_{t_{n-1}} \mid \mathcal{Y}_{t_{n-1}}] \big) + \frac{1}{2h}(Y_{t_n} - \mu h)^2}{Q_{t_n|t_{n-1}} + \frac{3}{2}}, \\
\operatorname{Var}(\nu_{t_n} \mid \mathcal{Y}_{t_n}) &= \frac{\left( \mathbb{E}[\nu_{t_n} \mid \mathcal{Y}_{t_n}] \right)^2 }{\left(Q_{t_n|t_{n-1}} + \frac{1}{2}\right)},
\end{aligned}
\end{equation}
\end{small}

\begin{small}
\noindent
where 
\[
Q_{t_n|t_{n-1}} = \frac{\big( \mathbb{E}[\nu_{t_{n-1}} \mid \mathcal{Y}_{t_{n-1}}] + h \, \mathbb{E}[\varphi_{t_{n-1}} \mid \mathcal{Y}_{t_{n-1}}] \big)^2}{\operatorname{Var}(\nu_{t_{n-1}} + \varphi_{t_{n-1}}\, h \mid \mathcal{Y}_{t_{n-1}}) + h \, \operatorname{Var}(\zeta(\nu_{t_{n-1}}) \mid \mathcal{Y}_{t_{n-1}})}.
\]
\end{small}

\noindent The functions \( \varphi\) and \( \zeta
\) are as defined in \eqref{eqn:discretedynamicsb}.
\end{theorem}

\begin{proof}[Proof of Theorem \ref{theorem:generallinkbetweenSVandPDV1}]
We begin with the correction step from Lemma \ref{lemma:alphabetadef}, which states that
\[
\alpha_{t_n|t_n} = \alpha_{t_n|t_{n-1}} + \frac{1}{2}, \quad \text{and} \quad 
\beta_{t_n|t_n} = \beta_{t_n|t_{n-1}} + \frac{1}{2h} (Y_{t_n} - \mu h)^2.
\]
Assuming that \(\nu_{t_n}\) follows an Inverse Gamma distribution, its conditional expectation is given by
\[
\mathbb{E}[\nu_{t_n} \mid \mathcal{Y}_{t_n}] = \frac{\beta_{t_n|t_n}}{\alpha_{t_n|t_n} - 1}.
\]
Substituting the updated parameters from the correction step, we obtain
\[
\mathbb{E}[\nu_{t_n} \mid \mathcal{Y}_{t_n}] = 
\left( \beta_{t_n|t_{n-1}} + \frac{1}{2h}(Y_{t_n} - \mu h)^2 \right) 
\left( \alpha_{t_n|t_{n-1}} - \frac{1}{2} \right)^{-1}.
\]
Expressing \(\alpha_{t_n|t_{n-1}}\) and \(\beta_{t_n|t_{n-1}}\) in terms of the moments provided in Lemma \ref{lemma:moments}, we obtain the conditional expectation in \eqref{eqn:corrandpred}.

For the conditional variance, since \(\nu_{t_n}\) is modelled as Inverse Gamma, we have
\[
\operatorname{Var}(\nu_{t_n} \mid \mathcal{Y}_{t_n}) = 
\frac{\beta_{t_n|t_n}^2}{(\alpha_{t_n|t_n} - 1)^2 (\alpha_{t_n|t_n} - 2)}.
\]
Substituting using Lemma \ref{lemma:alphabetadef}, this simplifies to
\[
\operatorname{Var}(\nu_{t_n} \mid \mathcal{Y}_{t_n}) = \frac{
\left( \mathbb{E}[\nu_{t_n} \mid \mathcal{Y}_{t_n}] \right)^2}{ 
\left( \Big[\alpha_{t_n|t_{n-1}} + \frac{1}{2}\Big] - 2 \right)}.
\]
Finally, substituting \(\alpha_{t_n|t_{n-1}}\) in Lemma \ref{lemma:moments} and \(Q_{t_n|t_{n-1}}\) as defined, we arrive at the desired result. 
\end{proof}
\begin{remark}  
This theorem establishes a recursive relationship between the expectations \(E[\nu_{t_n} | \mathcal{Y}_{t_n}]\) and \(E[\nu_{t_{n-1}} | \mathcal{Y}_{t_{n-1}}]\), which we can extend further to terms in \(E[\nu_{t_{n-j}} | \mathcal{Y}_{t_{n-j}}]\) for \(1 < j \leq n\) by the same principle. The recursion enables access to terms that involve past returns of the process, resembling a path-dependent volatility framework.  
\end{remark}  

\begin{corollary}
\label{corollary:path-dependent-volatility}
Define  
\(
A_{t_{n}|t_{n-1}} = Q_{t_n|t_{n-1}} + 1
\) and \(
B_{t_{n}|t_{n-1}} = Q_{t_n|t_{n-1}} + \frac{3}{2}
\),
where \( Q_{t_n|t_{n-1}} \) is as given in Theorem \ref{theorem:generallinkbetweenSVandPDV1}. Then, by setting  
\(
K_{t_{n}|t_{n-1}} = \frac{A_{t_{n}|t_{n-1}}}{B_{t_{n}|t_{n-1}}}
\),
it follows from Theorem \ref{theorem:generallinkbetweenSVandPDV1} that an iterative expansion of the terms in 
\( E[\nu_{t_{n-1}} \mid \mathcal{Y}_{t_{n-1}}] \)
recursively reveals \( E[\nu_{t_{n-m}} \mid \mathcal{Y}_{t_{n-m}}] \) for \( 1< m \leq n \), such that the conditional mean volatility is given by
\begin{small}
\begin{align*}
   \mathbb{E}[\nu_{t_n} \mid \mathcal{Y}_{t_n}] &= \left(\prod_{j=1}^{n-1} K_{t_j \mid t_{j-1}} \right) \nu_{t_0} 
    + \sum_{t_i \leq t_{n-1}} \frac{(Y_{t_i} - \mu h)^2}{2h B_{t_i|t_{i-1}}} \left( \prod_{j=i}^{n-1} K_{t_j \mid t_{j-1}} \right) 
    + \frac{1}{2h B_{t_n|t_{n-1}}} (Y_{t_n} - \mu h)^2 + G(\varphi),
\end{align*}
\(
\text{where } 
G(\varphi) = \sum_{t_i \leq t_{n-1}} h \, \mathbb{E}\left[\varphi_{t_i} \mid \mathcal{Y}_{t_i} \right] \prod_{j=i}^{n-1} K_{t_j \mid t_{j-1}}.
\)
\end{small}
\end{corollary}

\subsection{Accounting for the leverage effect in the ADF}
Volatility is typically negatively correlated with returns. For simplicity, we have so far assumed zero correlation. To incorporate this, we now consider Lemma \ref{lemma:leverageeffectonvoleulerscheme} below.
\begin{lemma}
Let the SDE system be given as in \eqref{eqn:svsystem}, and suppose that the Wiener increments driving the returns and volatility processes are correlated with coefficient \(\rho \in [-1,1]\). Then, under a discrete-time Euler--Maruyama scheme, the volatility process admits the following approximation,
\begin{equation}
\begin{split}
\nu_{t_n} = \, & \nu_{t_{n-1}} + \varphi(\nu_{t_{n-1}}) \, h 
 + \rho \, \frac{\zeta(\nu_{t_{n-1}})}{g(\nu_{t_{n-1}},X_{t_{n-1}})} 
\left((X_{t_n}-X_{t_{n-1}}) - u(\nu_t,X_{t_{n-1}})\, h\right) \\
& + \zeta(\nu_{t_{n-1}}) \sqrt{(1 - \rho^2)h}\, Z_{t_n}^{*},
\end{split}
\label{eqn:euler_volatility}
\end{equation}
where \(Z_{t_n}^{*}\) is a standard normal increment independent of \(X\) and \(\mathcal{Y}_{t_n}\). The functions \(\varphi\), \(g\), \(u\), and \(\zeta\) are defined in \eqref{eqn:svsystem}.
\label{lemma:leverageeffectonvoleulerscheme}
\end{lemma}
\begin{proof}
Suppose we are given the following SDE system 
\begin{equation}
\label{eqn:svsystem_redefined}
\begin{split}
dX_t &= u(\nu_t, X_t) \, dt + g(\nu_t,X_{t}) \, dW_t^i, \\
d\nu_t &= \varphi(\nu_t) \, dt + \zeta( \nu_t) \, dW_t^j,
\end{split}
\end{equation}
where the Wiener processes \(W^i\) and \(W^j\) are correlated, such that for a correlation coefficient \(\rho\), we have
\begin{equation}
\label{eqn:Wj}
dW_t^j = \rho \, dW_t^i + \sqrt{1 - \rho^2} \, dW_t^*,
\end{equation}
and \(W^*\) is a Brownian motion. From the first equation of the system in \eqref{eqn:svsystem_redefined}, we can formally isolate \(dW_t^i\) as
\begin{equation}
\label{eqn:Wi}
dW_t^i = \frac{dX_t - u(\nu_t, X_t)\, dt}{g(\nu_t,X_{t})}.
\end{equation}

Substituting \eqref{eqn:Wi} into \eqref{eqn:Wj}, and then replacing \(dW_t^j\) in the second equation of \eqref{eqn:svsystem_redefined}, we obtain the dynamics of the volatility process,
\[
d\nu_t = \varphi(\nu_t) \, dt + \rho \, \frac{\zeta(\nu_t)}{g(\nu_t,X_{t})} \left(dX_t - u(\nu_t,X_t)\, dt\right) + \zeta(\nu_t) \sqrt{(1 - \rho^2)} \, dW_t^*.
\]
Given the ratio \({\zeta(\nu_{t}})/{g(\nu_{t},X_{t})}\) is a real-valued function, applying the Euler–Maruyama discretization yields the discrete-time volatility approximation stated in \eqref{eqn:euler_volatility}.
\end{proof}
\section{Simple Example: Heston volatility model}

To illustrate the performance of the ADF model, in the remainder of this paper we will focus on an example based on the classic Heston stochastic volatility process. We show that the ADF model achieves an in-sample fit comparable to the PDV model proposed by Guyon and Lekeufack \cite{Guyon2023VOLATILITYGUYON}, while consistently outperforming in out-of-sample tests. In this section, we focus on calibrated simulations, in section 4.2 we see these results carry over to real data.

Suppose the Heston  dynamics for the stochastic volatility \(\nu\) are given by
\begin{equation}
\label{eqn:mean-revertingvol}
\begin{split}
dX_t &= \mu \, dt + \sqrt{\nu_t} \, dW_t^{i}, \\
d\nu_t &= \kappa(\theta - \nu_t) \, dt + \xi \sqrt{\nu_t} \, dW_t^{j},
\end{split}
\end{equation}
where \(\nu\) mean-reverts to the long-term level \(\theta\) at rate \(\kappa\), and \(\xi\) denotes the volatility of volatility. The Brownian motions \(W_t^{i}\) and \(W_t^{j}\) have correlation coefficient \(\rho \in [-1,1]\).

\begin{lemma}
\label{lemma:eulermaruyammeanrevertingprocess}
The discrete-time approximation (such as in \eqref{eqn:euler_volatility}) of the mean-reverting stochastic volatility model in \eqref{eqn:mean-revertingvol} is given by
\begin{equation}
\label{eqn:discreteheston}
\begin{aligned}
X_{t_n} &= X_{t_{n-1}} + \mu h + \sqrt{\nu_{t_{n-1} }h} \, Z^{i}_{t_n}, \\
\nu_{t_n} &= \nu_{t_{n-1}} + \kappa(\theta - \nu_{t_{n-1}}) h 
+ \rho\, \xi \left((X_{t_n} - X_{t_{n-1}}) - \mu h\right) 
+ \xi \sqrt{\nu_{t_{n-1}} (1 - \rho^2) h} \, Z^{*}_{t_n},
\end{aligned}
\end{equation}
where \(Z^i_{t_n}, Z^*_{t_n} \sim \mathcal{N}(0,1)\) are independent standard normal random variables.
\end{lemma}

\begin{remark}
We impose the Feller condition, \(2\kappa\theta >\xi^2\), which ensures the continuous-time volatility process remains strictly positive and does not hit zero almost surely. 

For the discrete-time volatility, since the Euler--Maruyama scheme can sometimes produce negative values for the variance, we apply a truncation step by replacing any negative volatility values with zero to ensure non-negativity.
\end{remark}

\begin{corollary}
 \label{corollary:ADF_dynamics}
   By applying the recursive relationship from Theorem \ref{theorem:generallinkbetweenSVandPDV1} to the result in Lemma \ref{lemma:eulermaruyammeanrevertingprocess}, we obtain the following expressions for the ADF under the Heston dynamics, 
\begin{small}
  \begin{align}
\mathbb{E}[\nu_{t_n} \mid \mathcal{Y}_{t_n}] 
&= \Bigg\{ 
    (Q_{t_n|t_{n-1}} + 1)(1 - \kappa h)\, \mathbb{E}[\nu_{t_{n-1}} \mid \mathcal{Y}_{t_{n-1}}]  + (Q_{t_n|t_{n-1}} + 1)\left(\kappa \theta h + \rho\, \xi \left(Y_{t_n} - \mu h\right)\right) \notag \\
&\quad + \frac{1}{2h}(Y_{t_n} - \mu h)^2 
\Bigg\} \left(Q_{t_n|t_{n-1}} + \frac{3}{2}\right)^{-1}, \label{eqn:varQADFeq}\\
\operatorname{Var}(\nu_{t_n} \mid \mathcal{Y}_{t_n}) 
&= \Bigg(\mathbb{E}[\nu_{t_n} \mid \mathcal{Y}_{t_n}]\Bigg)^2 
    \Bigg(Q_{t_n|t_{n-1}} + \frac{1}{2} \Bigg)^{-1}, \notag
\end{align}
\end{small}
where, \(Q_{t_n|t_{n-1}}\) is as defined in Theorem \ref{theorem:generallinkbetweenSVandPDV1}, that is,
\begin{align}
\label{eqn:QtnADF}
Q_{t_n|t_{n-1}} = \frac{\left( \mathbb{E}[\nu_{t_n} \mid \mathcal{Y}_{t_{n-1}}] \right)^2}
{\operatorname{Var}(\nu_{t_n} \mid \mathcal{Y}_{t_{n-1}})} 
= \frac{\left(\kappa \theta h + \rho\, \xi \left(Y_{t_n} - \mu h\right) + (1 - \kappa h)\mathbb{E}[\nu_{t_{n-1}} \mid \mathcal{Y}_{t_{n-1}}] \right)^2}
{(1 - \kappa h)^2 \operatorname{Var}(\nu_{t_{n-1}} \mid \mathcal{Y}_{t_{n-1}}) + \xi^2 (1 - \rho^2) \,h \, \mathbb{E}[\nu_{t_{n-1}} \mid \mathcal{Y}_{t_{n-1}}]}.
\end{align}
\end{corollary}
\begin{proof}
    The expression for the conditional expectation \( E[\nu_{t_n} \mid \mathcal{Y}_{t_{n-1}}] \) follows directly by taking the conditional expectation of the second equation in~\eqref{eqn:discreteheston}. To derive the conditional variance,
    \begin{align*}
        \operatorname{Var}\big(\nu_{t_n} \mid \mathcal{Y}_{t_{n-1}}\big)
        &= \operatorname{Var}\Big(\kappa \theta h + (1 - \kappa h)\nu_{t_{n-1}} +  \rho\, \xi (Y_{t_n} - \mu h) + \xi \sqrt{\nu_{t_{n-1}}  (1 - \rho^2)  h} \, Z^{j}_{t_n} \,\Big|\, \mathcal{Y}_{t_{n-1}}\Big) \\
        &= (1 - \kappa h)^2 \operatorname{Var}\big(\nu_{t_{n-1}} \mid \mathcal{Y}_{t_{n-1}}\big) \\
        &\quad + \xi^2 (1 - \rho^2) h\left( \mathbb{E}\left[\left(\sqrt{\nu_{t_{n-1}} } \, Z^{j}_{t_n} \right)^2 \,\big|\, \mathcal{Y}_{t_{n-1}} \right] 
        - \left(\mathbb{E}\left[\sqrt{\nu_{t_{n-1}} } \, Z^{j}_{t_n} \,\big|\, \mathcal{Y}_{t_{n-1}}\right]\right)^2 \right).
    \end{align*}
\noindent Acknowledging that \( \mathbb{E}[Z^{j}_{t_n} \mid \mathcal{Y}_{t_{n-1}}] = 0 \) and \( \mathbb{E}[(Z^{j}_{t_n})^2 \mid \mathcal{Y}_{t_{n-1}}] = 1 \), since \( Z^{j}_{t_n} \sim \mathcal{N}(0,1) \) and is independent of \( \nu_{t_{n-1}} \), it follows that
    \begin{equation}\label{eq:conditionalvariancecalculation}
        \operatorname{Var}(\nu_{t_n} \mid \mathcal{Y}_{t_{n-1}}) 
        = (1 - \kappa h)^2 \operatorname{Var}(\nu_{t_{n-1}} \mid \mathcal{Y}_{t_{n-1}}) 
        + \xi^2(1 - \rho^2)  h\, \mathbb{E}[\nu_{t_{n-1}} \mid \mathcal{Y}_{t_{n-1}}],
    \end{equation}
\noindent the result then follows by substituting these conditional expectation and conditional variance expressions into equation~\eqref{eqn:QtnADF}.
\end{proof}
\begin{remark}
An algorithmic representation of the ADF is provided in Appendix~\ref{appendix:ADFalgorithm}.
\end{remark}

\subsection{ADF predictions in simulations}
The Assumed Density Filter can only see the daily returns data and does not have access to the underlying continuous volatility process. It estimates the continuous volatility based on discrete-time daily return observations. 

\begin{figure}[H]
    \centering
        \includegraphics[width=0.88\linewidth]{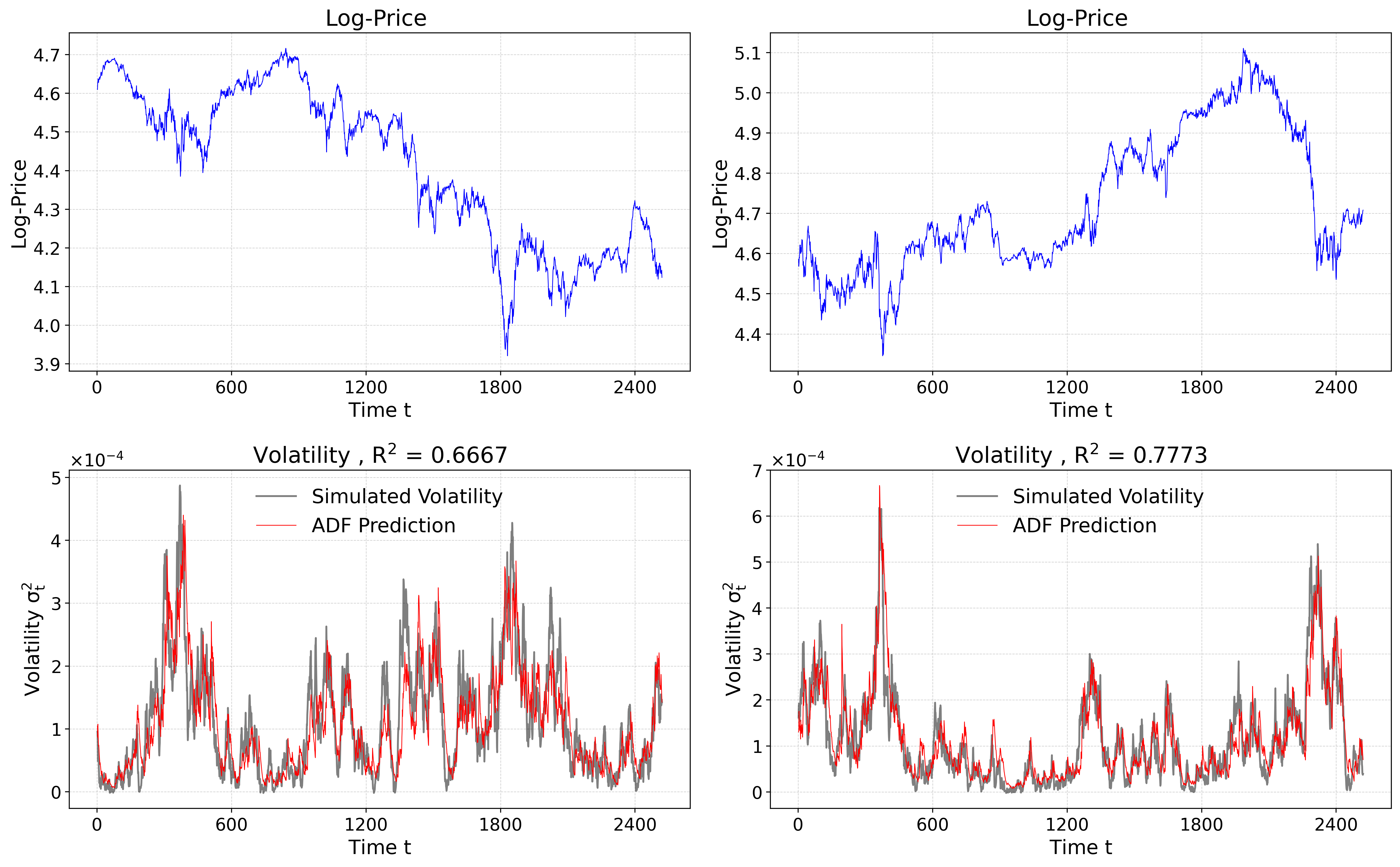}
  \caption{Heston model simulations with parameters from Gatheral \cite{GatheralTheGuide}.}
\label{fig:adfpredictionsinandoutofsample}
\end{figure}

To evaluate the best-case performance of the ADF, we use Heston parameters consistent with the ranges proposed by Gatheral \cite{GatheralTheGuide},  both as the basis for simulation and to run our ADF.  Specifically, we set the annualized parameters to \(\theta = 0.035\), \(\kappa = 2.75\), \(\xi = 0.425\), \(\rho = -0.4644\), and \(\mu = 0.05\). The parameter set satisfies the Feller condition, guaranteeing non-negative variances, and thus serves as a convenient yet representative choice for illustration.

Figure~\ref{fig:adfpredictionsinandoutofsample} illustrates two examples of the performance of the ADF when estimating the volatility simulated by the Euler--Maruyama Scheme for the Heston stochastic volatility model. The corresponding \(R^2\) values are in the chart titles. We observe that the model achieves high \(R^2\) values on the simulated data.

\subsubsection{Approximating \(Q_{t_n|t_{n-1}}\) by the expectation \(\mathbb{E}[Q_{t_n|t_{n-1}}]\)}
A key observation is that assuming a constant value for \( Q_{t_n|t_{n-1}} \) reduces the ADF model to a  time-weighted-average form similar to that proposed by Guyon and Lekeufack \cite{Guyon2023VOLATILITYGUYON}. 

\begin{lemma}
\label{lemma:pdvfromADF}
Define  
\(
A_{t_{n-1}} = (Q_{t_n|t_{n-1}} + 1)(1-\kappa h) \) and \( B_{t_{n-1}} = Q_{t_n|t_{n-1}} + \frac{3}{2}
\), where \( Q_{t_n|t_{n-1}} \) is defined as in Theorem~\ref{theorem:generallinkbetweenSVandPDV1}. Suppose \( Q_{t_n|t_{n-1}} \) can be approximated by a constant (i.e., \( Q_{t_n|t_{n-1}} = Q \) for all \( t_n \)), it follows that \( K_{t_n|t_{n-1}} = \frac{A_{t_n|t_{n-1}}}{B_{t_n|t_{n-1}}} \) is also approximately a constant \(K\). The ADF formulation in Corollary \ref{corollary:path-dependent-volatility} reduces to
\begin{equation}
\begin{split}
\mathbb{E}[\nu_{t_n} \mid \mathcal{Y}_{t_n}] &= K^{t_n - t_0} \nu_{t_0}  + \frac{1}{2B} \left( \mu^2 h + 2(Q + 1)(\kappa\theta - \rho \xi \mu) h \right) \sum_{t_i \leq t_n} K^{t_n - t_i}\\
&\quad + \frac{1}{B} ((Q+1)\rho \xi-\mu) \sum_{t_i \leq t_n} K^{t_n - t_i} Y_{t_i} + \frac{1}{2hB} \sum_{t_i \leq t_n} K^{t_n - t_i} Y_{t_i}^2,
\end{split}
\label{eq:expected_nu}
\end{equation}
    a path-dependent model of time-weighted-average type. See Appendix \ref{appendix:constPDVderivation} for details of the derivation.
\end{lemma}

\begin{figure}[H]
    \centering
        \includegraphics[width=\linewidth]{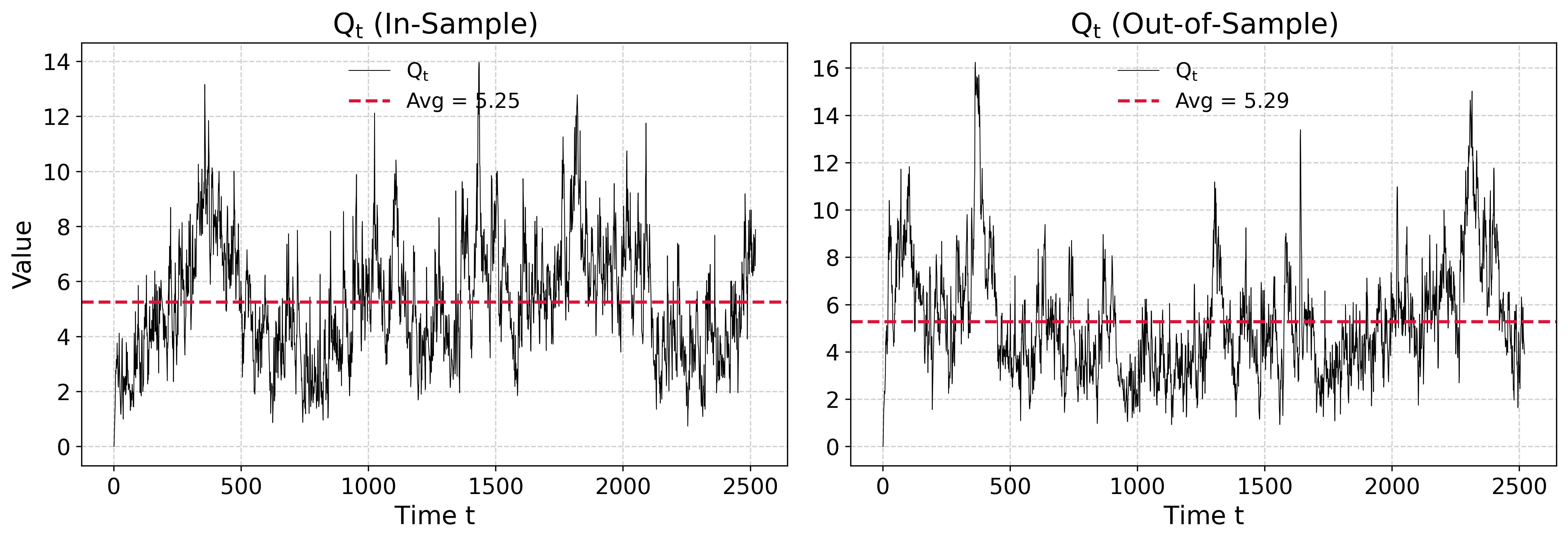}
    \caption{Empirical evidence for the stability of \( Q_{t_n|t_{n-1}} \) under the theoretical model.}
    \label{fig:egordicQ}
\end{figure}

Figure~\ref{fig:egordicQ} shows the time-varying values of \( Q_{t_n|t_{n-1}} \) computed through the ADF, using the true parameters from the underlying Heston volatility simulation illustrated in Figure~\ref{fig:adfpredictionsinandoutofsample}. Although \( Q_{t_n|t_{n-1}} \) is not  constant, it empirically shows some stability around its mean (as seen in Figure~\ref{fig:egordicQ}), consistent with ergodic-like behavior. This supports the simplification of approximating the dynamics by using a constant \( Q \). By Lemma~\ref{lemma:pdvfromADF}, this helps explain the strong performance of path-dependent models of the time-weighted average type, such as Guyon’s model. These observations are robust across multiple simulations (not shown here).

To evaluate the constant-\( Q \) simplification, we approximate \( Q \) by the in-sample mean of the estimated \( Q_{t_n|t_{n-1}} \) values, where ``in-sample" refers to the data shown in the left plot of Figure~\ref{fig:egordicQ}. The remaining portion of the data, shown in the right plot of Figure~\ref{fig:egordicQ} and excluded from the mean estimation, is treated as out-of-sample for evaluating the constant-\( Q \) model. We observe that the ADF model based on a constant value of \(Q\) has similar qualitative behaviour to a time-varying \(Q\) model, as we see in Figure \ref{fig:comparisonofconstandvarQadfpredictionsinandoutofsample}.

\begin{figure}[H]
    \centering    \includegraphics[width=0.9\linewidth]{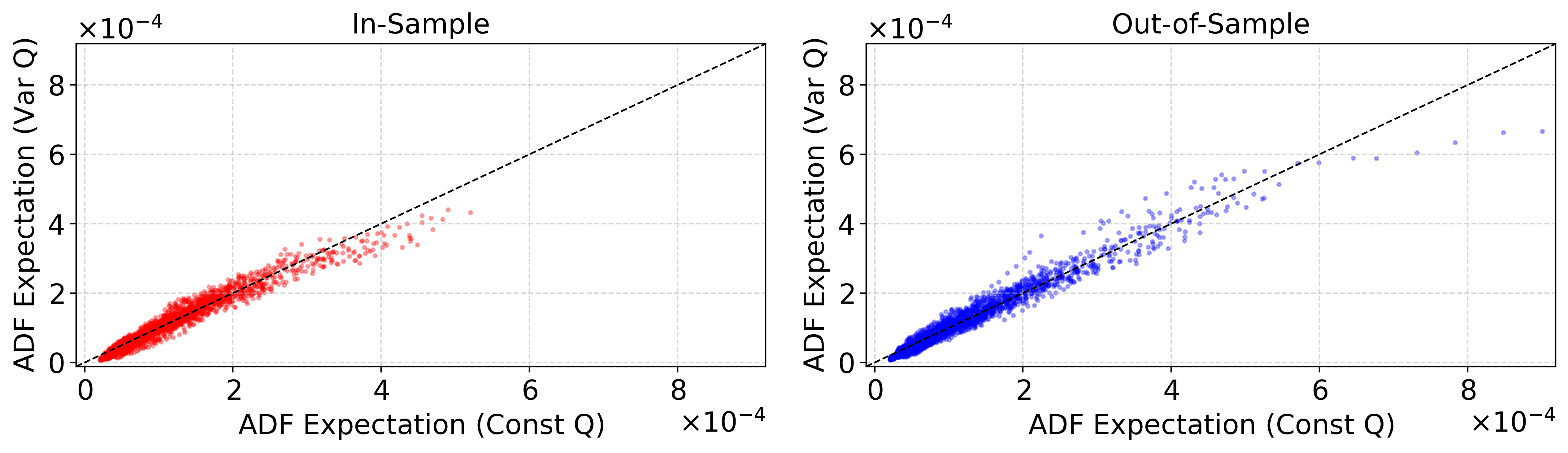}
    \caption{Scatter plot comparing volatility estimates from the ADF model with constant \(Q\) to those from the ADF model using the iteratively updated \(Q_{t_n}\).}
\label{fig:comparisonofconstandvarQadfpredictionsinandoutofsample}
\end{figure}
\noindent On the in-sample data, the constant-\(Q\) model achieves an \(R^2\) of \(0.6028\), compared to \(0.6387\) for the variable-\(Q\) specification. For the out-of-sample data, the corresponding \(R^2\) values are \(0.6299\) and \(0.6514\), respectively. 

The slight underperformance of the constant-\(Q\) model seems to be due to its underestimates of extreme volatility, particularly at higher levels. This behavior is consistent with the right-skewed distribution of \( Q_{t_n|t_{n-1}} \) (see Appendix \ref{appendix:distributionofqt}) since it is a strictly positive variable, leading to occasional large values that a constant approximation for \( Q_{t_n|t_{n-1}} \) may not represent well due to its assumption of symmetry in distribution.

\subsubsection{Comparison with traditional PDV models}
Figure~\ref{fig:correspondingpdvfit} compares our ADF model with a traditional path-dependent volatility model of the form
\begin{equation}\label{eq:PDVbasic}
\sigma^2_{t_n} = \beta_0 + \beta_1 \sum_{t_i \leq t_n} \lambda e^{-\lambda (t_n - t_i)} Y_{t_i} + \beta_2 \sum_{t_i \leq t_n} \lambda e^{-\lambda (t_n - t_i)} Y^2_{t_i}.
\end{equation}

We do this over 18 independent simulations (see Figure \ref{fig:correspondingpdvfit}). Overall, although all models demonstrate strong predictive performance, the ADF variants typically deliver modest but consistent improvements in \(R^2\).

\begin{figure}[H]  
\centering
\includegraphics[height=0.3\linewidth]{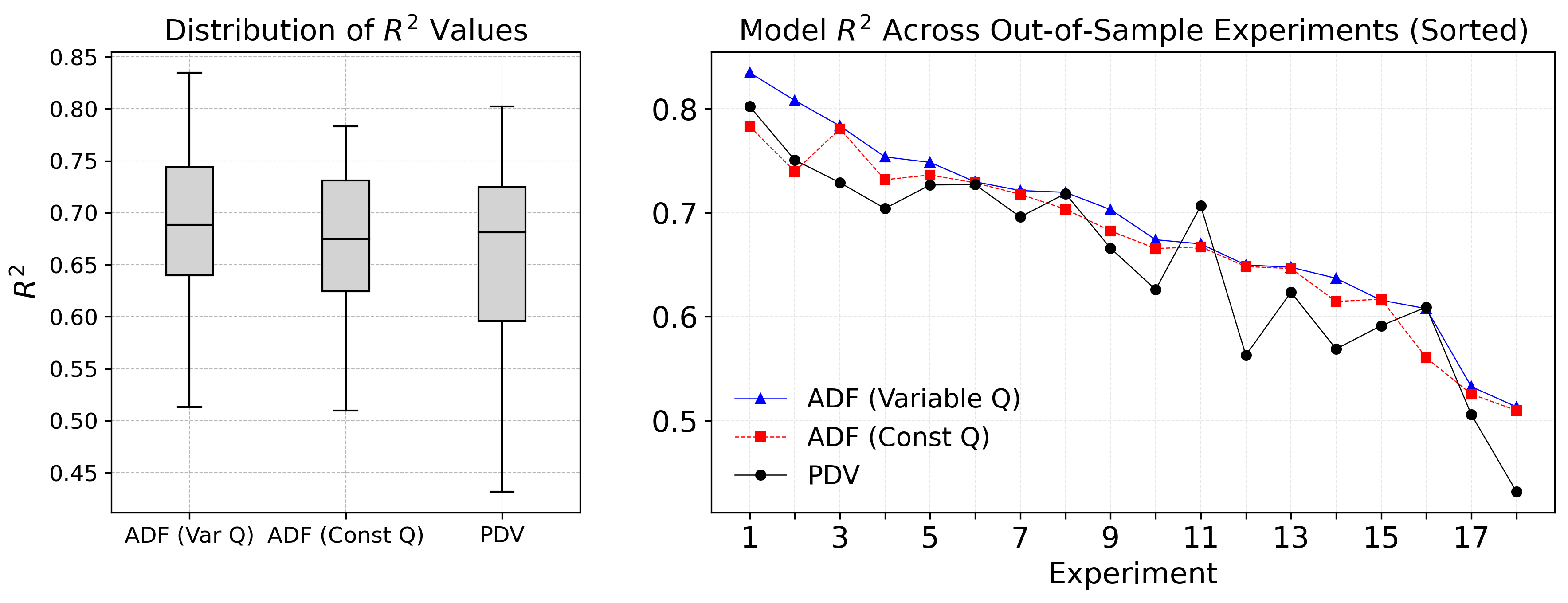}
\caption{Comparison of \(R^2\) values across models over multiple experiments. The results are sorted in descending order based on the \(R^2\) values of the variable-\(Q\) ADF model.}   \label{fig:correspondingpdvfit}
\end{figure}

\subsection{Estimating Heston parameters}

Calibrating the Heston model under the physical (\(\mathbb{P}\)) measure is challenging because volatility is not directly observable. Azencott, Ren, and Timofeyev \cite{Azencott2020RealizedSDEs} propose a moment-based approach that infers parameters from the empirical moments of realized volatility -- e.g. the Oxford--Man Institute’s 5-minute S\&P 500 realized volatility by Heber, Lunde, Shephard, and Sheppard~\cite{Heber2009} -- treating it as a proxy for the true latent process. This method requires large datasets and relies on asymptotic convergence, both of which may be impractical in application. Traditional maximum likelihood estimation (MLE) methods are sensitive to noise, yielding unstable estimates (for example, see, Cui, Ba\~no Rollin, and Germano \cite{Cui2016FullModel}).

Calibration difficulties also arise under the risk-neutral (\(\mathbb{Q}\)) measure. When calibrated to option prices, the Heston model often requires frequent re-estimation, despite theoretical expectations of parameter stability within a given volatility regime. This repeated calibration can itself induce parameter instability due to identifiability and other issues leading to non-unique solutions. For example, the long-term mean variance \(\theta\) and mean-reversion rate \(k\) can offset each other, producing distinct yet equally plausible fits \cite{Cui2016FullModel}. Such degeneracies create calibration risk, where alternative procedures can lead to materially different exotic option prices and hedge ratios (see Guillaume and Schoutens \cite{Guillaume2012CalibrationModel}).

The ADF provides a practical method for estimating these models under both the \(\mathbb{P}\) and \(\mathbb{Q}\) measures, using realized volatility or VIX data as noisy information sources. Crucially, it avoids the need to rely on risk-neutral option prices. Guyon and Parent \cite{GuyonThe} discuss the calibration of their discrete time 4-factor model under the \(\mathbb{P}\) and \(\mathbb{Q}\) measures, we extend their discussion to the Heston volatility model. 

\subsubsection{PDV-Based estimation of Heston parameters}
We might ask, given the close connection we have explored between SV and PDV models, whether it is possible to directly infer the parameters of an SV model from a fitted PDV model. 

Focusing on the Heston model, we ask whether we can find a map $(\beta_0, \beta_1, \beta_2, \lambda)\to (\rho, \kappa, \theta, \xi)$. A na{\"\i}ve approach comes from  comparing the coefficients in \eqref{eq:expected_nu} and \eqref{eq:PDVbasic}. In particular,  under the simplifying assumption that $t_0\approx -\infty$ and $\mu\approx 0$, and our observations are for $t_n \in \mathbb{Z}$. If we match coefficients between the PDV and ADF, we obtain the relationships

\begin{align*}
\frac{(Q+1)(\kappa\theta)h}{(1-K)B} &= \beta_0, \qquad \frac{(Q+1)\rho\xi}{B} = \lambda \beta_1,\qquad \frac{1}{2hB} = \lambda \beta_2, \qquad K = e^{-\lambda}, 
\end{align*}
together with the approximation $B=Q+\frac{3}{2}$ as in Lemma \ref{lemma:pdvfromADF}. If we also know the residual variance $s^2$ of the PDV model, we can use the time-homoskedastic approximation of \eqref{eq:conditionalvariancecalculation} to obtain
\[\mathrm{Var}(\nu_{t_n}|\mathcal{Y}_{t_{n-1}}) = s^2 \approx (1-\kappa h)^2 s^2 +\xi^2(1-\rho^2)h \theta\]
(replacing the expectation with its long-term average), which allows us to separate $\rho$ and $\xi$. Table \ref{table:PDVADFequivalentparams} summarizes this connection between the coefficients of the PDV and ADF models, together with their estimates when using the PDV fitted to the path in Figure \ref{fig:adfpredictionsinandoutofsample} (left).

\begin{table}[h!]
\centering
\begin{tabular}{c r |c r}
\hline
PDV coeff. & Fitted values& ADF coeff. & True values \\
\hline
$ \beta_0$ & $2.362\times 10^{-5}$ &$\frac{(Q+1)(\kappa\theta)h}{(1-K)B}$& $1.553\times 10^{-6}$\\
$\lambda \beta_1$ &  $-0.823\times 10^{-3}$ & $\frac{(Q+1)\rho\xi}{B}$ & $-7.553\times 10^{-3}$ \\
$\lambda \beta_2$ &  0.0505& $\frac{1}{2hB}$&0.0723  \\
$e^{-\lambda}$ & 0.9374& $K=\frac{(Q+1)(1-\kappa h)}{B}$ & 0.9176 \\
\hline
$s^2$& $2.4456\times 10^{-9}$&$\frac{\xi^2(1-\rho^2)\theta}{\kappa(2-\kappa h)}$& $1.42744\times 10^{-8}$\\
\hline
\end{tabular}
\caption{A comparison of the PDV’s fitted values with the corresponding true values in the ADF (calculated from the underlying Heston parameters in the simulation).}
\label{table:PDVADFequivalentparams}
\end{table}

Rearranging these expressions, we can also attempt to infer the parameters of a Heston model which correspond to the given fitted PDV coefficients. We first compute the auxiliary constants:
\[
K = 0.9374, 
\qquad 
B = \frac{1}{2h\lambda\beta_2} = 9.8973, 
\qquad 
Q = B - \tfrac{3}{2} = 8.3973.
\]
From these, we obtain
\[
\kappa = \frac{1}{h}\left(1 - \frac{KB}{Q+1}\right) = 0.01272,
\qquad 
\theta = \frac{\beta_0 (1-K)B}{(Q+1)\kappa h} = 1.2245 \times 10^{-4}.
\]
Finally, the volatility parameters are given by
\begin{align*}
\xi &= \sqrt{\frac{\kappa(2-\kappa h)}{\theta} s^2 
+ \left(\frac{B\lambda\beta_1}{Q+1}\right)^2}
= 1.094 \times 10^{-3},\\
\rho &= \frac{B\lambda \beta_1}{(Q+1)\xi} = -0.7925.
\end{align*}

We compare these crude PDV-based estimates with their true values in Table \ref{table:PDVADFequivalentHestonparams}.

\begin{table}[h!]
\centering
\begin{tabular}{c | r r}
\hline
 & True value& PDV-implied value \\
\hline
$\kappa$  & $0.0109$                & $0.01272$\\
$\theta$  & $1.389\times 10^{-4}$ &$1.224\times 10^{-4}$   \\
$\xi$     & $1.687\times 10^{-3}$   & $1.094\times 10^{-3}$\\
$\rho$    & $-0.4644$               &$-0.7925$\\
\hline
\end{tabular}
\caption{Heston parameters and their approximate implied values from the corresponding fitted PDV model.}
\label{table:PDVADFequivalentHestonparams}
\end{table}
As we see, in this case the estimates of $\kappa$ and $\theta$ in this case are moderately accurate, however there is noticeable error in the estimates of $\xi$ and $\rho$. That both of these are subject to error is unsurprising, given they are closely related quantities, and there is a clear discrepancy in the fitted value of $\beta_1$ and the corresponding ADF coefficient in Table \ref{table:PDVADFequivalentparams}.

\subsubsection{ADF-Based estimation of Heston parameters}

The previous approach to calibrating a Heston model is deliberately crude, as it only uses a fitted PDV model, corresponding to a constant-Q approximation of the ADF. A more nuanced approach exploits the true ADF, to obtain a filtering-based estimate of Heston parameters. 

We consider the same simulations as previously given in Figure~\ref{fig:adfpredictionsinandoutofsample} (left). Our objective is to recover the true parameters generating the volatility process by fitting the variable-\(Q\) ADF model (as in \eqref{eqn:varQADFeq}, and parameterized using the Heston parameters) to the training data (Figure~\ref{fig:adfpredictionsinandoutofsample}: left plot) through optimization, given access only to returns data. Calibration is carried out by minimizing the residuals between predicted and realized in-sample volatility. The model’s performance is then evaluated on the out-of-sample test data (Figure~\ref{fig:adfpredictionsinandoutofsample}: right plot).

\begin{table}[H]
\centering
\label{tab:heston_params}
\renewcommand{\arraystretch}{}
\begin{tabular}{c|c|c|c}
\textbf{Parameter} & \textbf{True Value} & \textbf{Estimate} & \textbf{Approx. Std. Error} \\
\hline
\(\kappa\)   & 0.0109   & 0.01208  & 0.006238\\
\(\theta\)      & \(1.389\times 10^{-4}\)  &  \(1.048\times 10^{-4}\)  &  0.4382  \\
\(\xi\)     & $1.687\times 10^{-3}$    & $1.947\times 10^{-3}$ &$2.640 \times 10^{-4}$ \\
\(\rho\)   & -0.4644 & -0.37685 &0.1171   \\
\hline
\end{tabular}
\caption{Heston parameter estimates with standard errors approximated as discussed in Section~4.1.4.}
\end{table}
The model calibrated with optimized parameters achieves comparable \(R^2\) values: \(0.6547\) in-sample versus \(0.6387\) for the ADF fitted using the true parameters of the underlying process. Out-of-sample, the optimized model attains \(0.6607\) compared to \(0.6514\) for the true ADF. We can see that these estimates are a noticeable improvement over the estimation of parameters using the simpler PDV approximation, as discussed in Table \ref{table:PDVADFequivalentHestonparams}.

\subsection{ADF-based estimation of S\&P 500 realized volatility}
We now apply the ADF-Based estimation method to real data, using the dynamics in Corollary~\ref{corollary:ADF_dynamics}. The model is fitted to realized volatility computed from 5-minute returns provided by the Oxford--Man Institute, covering 2000–2018. We use 2000–2009 as the in-sample period for estimation and 2009–2018 as the out-of-sample period for evaluation. At each step, daily S\&P~500 returns up to time \(t_n\) are used to predict realized volatility over \([t_n, t_{n+1})\).

\begin{table}[H]
\centering
\renewcommand{\arraystretch}{}
\begin{tabular}{c|c|c}
\textbf{Parameter}  & \textbf{Estimate} & \textbf{Approx. Std. Error} \\
\hline
\(\kappa\)       & 0.07908   & 0.03503   \\
\(\theta\)       & $4.123 \times 10^{-5}$ & 0.1922  \\
\(\xi\)          & $5.105\times 10^{-3}$& $1.554\times 10^{-3}$\\
\(\rho\)         & -0.4784   & 0.1537  \\
\hline
\end{tabular}
\caption{Standard errors are computed using the parametric bootstrap method. The corresponding annualized parameters are: \(\kappa  = 19.93\), \(\xi = 1.29\), \(\theta = 0.0104\), with \(\rho\) unchanged.}
\label{tab:heston_params2}
\end{table}

\begin{figure}[H]
    \centering
    \begin{subfigure}[b]{\linewidth}
        \centering        \includegraphics[width=\linewidth]{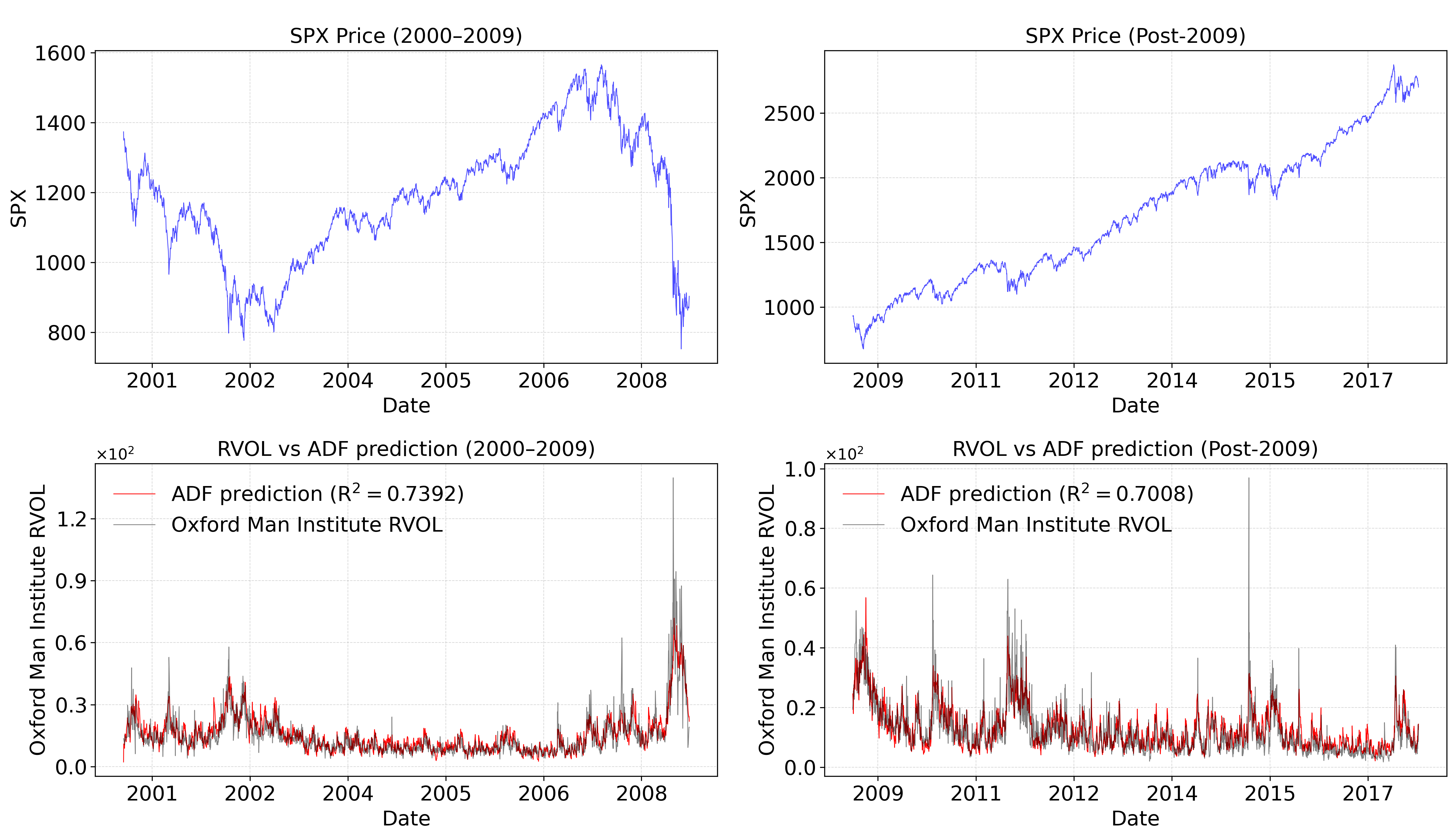}        
        \label{fig:ADFRVOLplot_combined}
    \end{subfigure}
     \caption{In-sample and out-of-sample realized volatility predictions using the ADF model. The top row shows the S\&P 500 log prices, and the bottom row displays the corresponding realized volatility estimates.
}
    \label{fig:ADFmodel_combined}
\end{figure}

\begin{figure}[H]
        \centering
\includegraphics[width=0.85\linewidth]{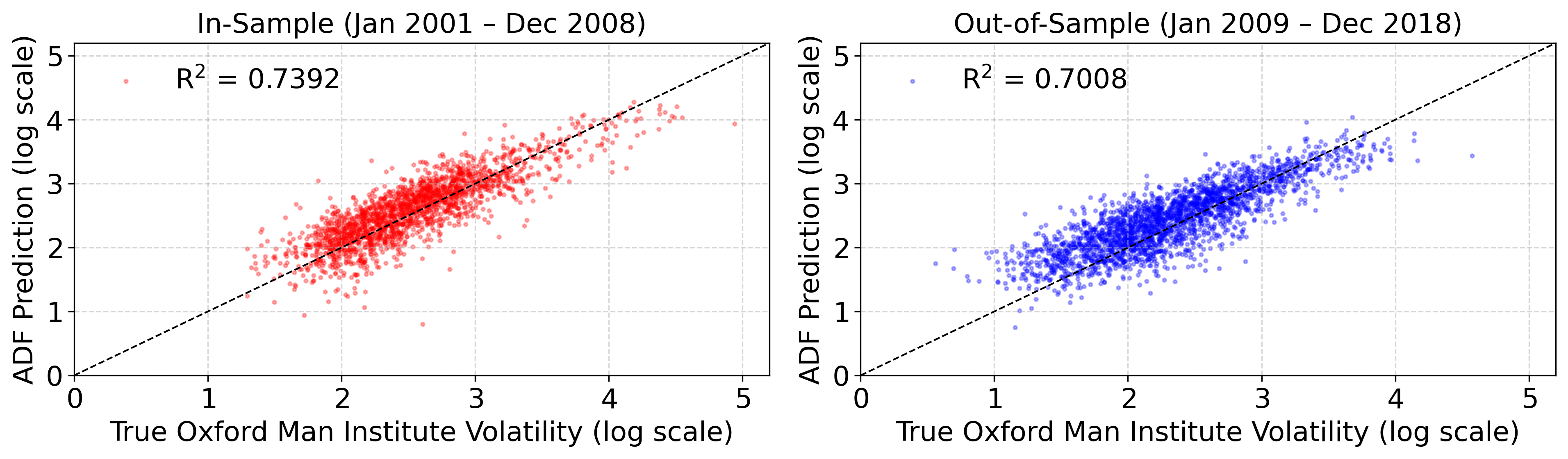}
 \caption{Plot of realized volatility (true) versus ADF model predictions.
} \label{fig:ADFResiduals_combinedscatter}
    \end{figure}
Figures~\ref{fig:ADFmodel_combined} and \ref{fig:ADFResiduals_combinedscatter} present the performance of the model. In-sample, the model achieves \(R^2 = 0.7397\), comparable to \(R^2 \approx 0.738\) reported by Guyon and Lekeufack \cite{Guyon2023VOLATILITYGUYON}. When applied to the out-of-sample period, the model maintains strong predictive accuracy with \(R^2 = 0.70\) (Guyon and Lekeufack \cite{Guyon2023VOLATILITYGUYON} achieves \(R^2 \approx 0.65\)).  The model’s out-of-sample performance is after 2009 and despite the shift to a different volatility regime following the Financial Crisis, it remains strong. We attribute the remaining unexplained variation to exogenous randomness inherent in the volatility process. 
\subsection{ADF-based estimation of VIX}
We also fit the ADF model to VIX data, using past S\&P~500 daily returns up to \(t_n\) to predict the VIX at \(t_{n+1}\).
\begin{table}[H]
\centering
\renewcommand{\arraystretch}{}
\begin{tabular}{c|c|c}
\textbf{Parameter} & \textbf{Estimate} & \textbf{Approx. Std. Error} \\
\hline
\(\kappa\)   & 0.01768    & 0.005401   \\
\(\theta\)   & $1.316\times 10^{-4}$ & 0.2073  \\
\(\xi\)      & $1.325\times 10^{-3}$  & 0.0001484   \\
\(\rho\)     & -0.6711    & 0.08244  \\
\hline
\end{tabular}
\caption{These values differ slightly but remain within the range typically suggested for the Heston model in the literature \cite{GatheralTheGuide,Zhang2003PricingModel}.}
\label{tab:vix_model_params}
\end{table}

\begin{figure}[H]
    \centering
\includegraphics[width=\linewidth]{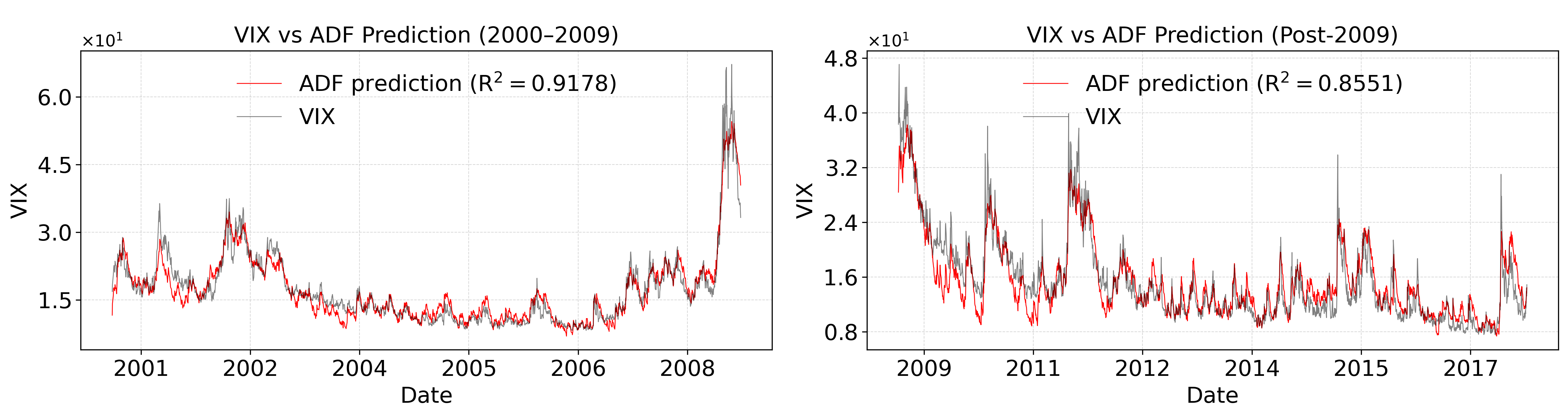}
    \caption{In-sample and out-of-sample implied volatility (VIX) predictions using the ADF model. }
    \label{fig:ADFVIXplot2}
\end{figure}

We obtain an in-sample \(R^2 = 0.92\), compared to \(R^2 = 0.946\) reported by Guyon and Lekeufack \cite{Guyon2023VOLATILITYGUYON}. Extending the analysis to the out-of-sample period, our model achieves 
\(R^2 = 0.85\) (similarly, \cite{Guyon2023VOLATILITYGUYON} achieves \(R^2 \approx 0.855\)).

\begin{figure}[H]
    \centering
\includegraphics[width=0.9\linewidth]{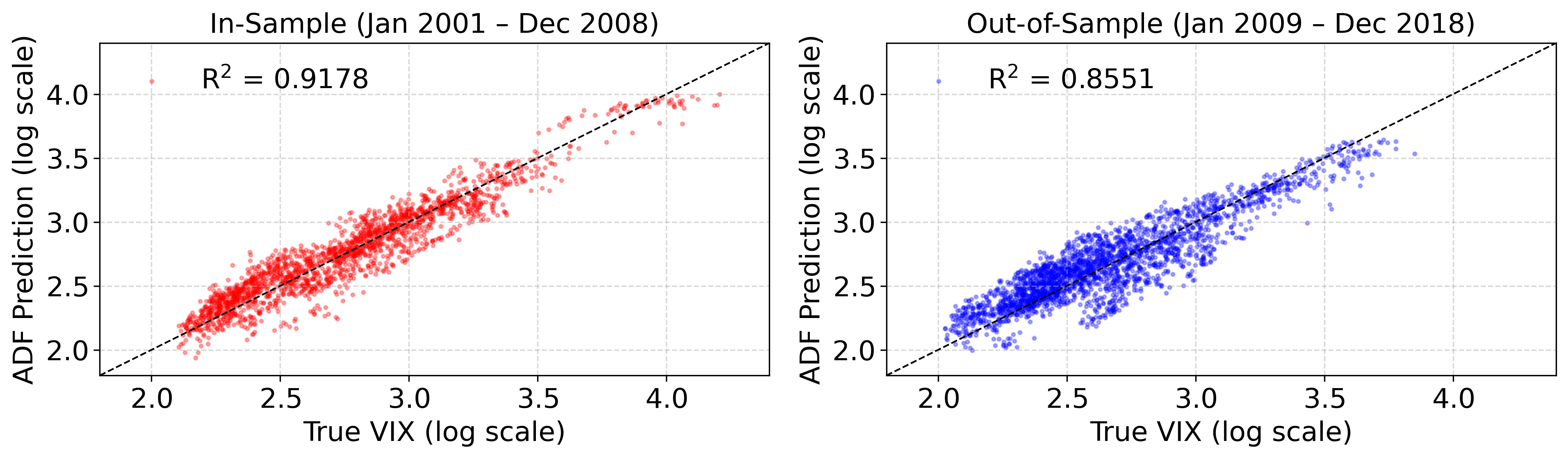}
    \caption{Plot of VIX (true) versus ADF model predictions.}
    \label{fig:ADFVIXplot21}
\end{figure}

 Although the model has access only to returns data, it nonetheless captures the dynamics of the VIX with a strong fit.
 
\section{Conclusion}
This study explored a link between stochastic volatility and path-dependent volatility models via the Assumed Density Filter. Through a formal and theoretically grounded argument, we showed that PDV models derived from SV dynamics can replicate key features of market volatility -- capturing both the cross-sectional distribution of latent volatility and its dependence on the realized historical returns path.

Building on this connection, we proposed a calibration procedure for both SV and PDV models by optimizing the ADF on in-sample data, achieving high accuracy both in-sample and out-of-sample. Notably, the choice of initial values has no lasting impact due to self-correcting/contractive dynamics of the filter. This framework combines the tractability and calibration efficiency of PDV models with the data-generating capabilities of SV models through Monte Carlo simulation. It supports joint calibration to the S\&P~500 and VIX smiles and extends PDV to fast, realistic synthetic data generation and derivative pricing.

In summary, the ADF-based PDV framework provides a rigorous and practical approach to volatility modeling. Its unification of theoretical consistency with empirical robustness makes it a valuable tool for both academic research and industry applications in derivative pricing and risk management. Future research may explore the generation of synthetic data based on the ADF model, which arises as the conditional expectation of stochastic volatility.

\section*{Acknowledgements}
The authors gratefully acknowledge funding support from the Rhodes Trust and the School of Geography and the Environment, University of Oxford. We also extend our sincere thanks to Doyne Farmer, Mary Eaton, Maureen Freed, Joel Dyer, François Lafond, Samuel Wiese, Dorothy Nicholas, and Edward Nicholas for their valuable support and contributions.

\printbibliography

\section*{Appendices}
\begin{appendix}

\section{Proof of Lemma \ref{lemma:numericalmoments}}
\label{appendix1:weak_convergence_euler_maruyama}
\begin{proof}
    We begin with the numerical solution of the general stochastic volatility differential equation as described in \eqref{eqn:discretedynamicsb}. Since \( Z_n^{j} \sim \mathcal{N}(0, 1) \), we proceed by calculating the conditional expectation and variance. For the expectation, we have,
    \begin{align*}
        \mathbb{E}\Big(\nu_{t_n} | \mathcal{Y}_{t_{n-1}}\Big) &= \mathbb{E}\left(\nu_{t_{n-1}} + h \varphi_{t_{n-1}} + \zeta_{t_{n-1}} \sqrt{h} Z_n^{j}\Big| \mathcal{Y}_{t_{n-1}}\right) \\
        &= \mathbb{E}\Big(\nu_{t_{n-1}}| \mathcal{Y}_{t_{n-1}}\Big) + h \mathbb{E}\Big(\varphi_{t_{n-1}} | \mathcal{Y}_{t_{n-1}}\Big) + \sqrt{h} \mathbb{E}\Big(\zeta_{t_{n-1}} | \mathcal{Y}_{t_{n-1}}\Big) \mathbb{E}\Big(Z_n^{j} | \mathcal{Y}_{t_{n-1}}\Big).
    \end{align*}
The random variable \(Z_n^{j}\) follows a standard normal distribution, satisfying \(
\mathbb{E}[Z_n^{j} | \mathcal{Y}_{t_{n-1}}] = 0
\). Consequently, the last term vanishes, leaving us with
    \begin{align*}
        \mathbb{E}\Big(\nu_{t_n} | \mathcal{Y}_{t_{n-1}}\Big) &= \mathbb{E}\Big(\nu_{t_{n-1}} | \mathcal{Y}_{t_{n-1}}\Big) + h \mathbb{E}\Big(\varphi_{t_{n-1}} |\mathcal{Y}_{t_{n-1}}\Big).
    \end{align*}
    Now, for the variance, we proceed as follows,
    \begin{align*}
        \operatorname{Var}\Big(\nu_{t_n}| \mathcal{Y}_{t_{n-1}} \Big) &= \operatorname{Var}\Big(\nu_{t_{n-1}} + h \varphi_{t_{n-1}} + \zeta_{t_{n-1}} \sqrt{h} Z_n^{j} \Big| \mathcal{Y}_{t_{n-1}}\Big) \\
        &= \operatorname{Var}\Big(\nu_{t_{n-1}}| \mathcal{Y}_{t_{n-1}}\Big) + h \operatorname{Var}\Big(\zeta_{t_{n-1}}| \mathcal{Y}_{t_{n-1}}\Big) \operatorname{Var}\Big(Z_n^{j}| \mathcal{Y}_{t_{n-1}}\Big) \\
        &= \operatorname{Var}\Big(\nu_{t_{n-1}}| \mathcal{Y}_{t_{n-1}} \Big) + h \operatorname{Var}\Big(\zeta_{t_{n-1}}| \mathcal{Y}_{t_{n-1}}\Big).
    \end{align*}
The key observation is that \( \operatorname{Var}\Big(Z_n^{j}| \mathcal{Y}_{t_{n-1}}\Big) = 1 \), as \( Z_n^{j} \) is standard normal. Therefore, we omit the variance of \( Z_n^{j} \), simplifying the expression.
\end{proof}

\section{Stepwise derivation of \eqref{eq:expected_nu}, as shown in Lemma \ref{lemma:pdvfromADF}}
\label{appendix:constPDVderivation}
Suppose we define  
\(
A_{t_{n-1}} = (Q_{t_n|t_{n-1}} + 1)(1-\kappa h) \) and \( B_{t_{n-1}} = Q_{t_n|t_{n-1}} + \frac{3}{2}
\), where \( Q_{t_n|t_{n-1}} \) is defined as in Theorem~\ref{theorem:generallinkbetweenSVandPDV1}. If we approximate \( Q_{t_n|t_{n-1}} \) by a constant (i.e., \( Q_{t_n|t_{n-1}} = Q \) for all \( t_n \)), it follows that \( K_{t_n|t_{n-1}} = \frac{A_{t_n|t_{n-1}}}{B_{t_n|t_{n-1}}} \) is also approximately a constant \(K\). We can rearranged the ADF formulation in Corollary \ref{corollary:path-dependent-volatility} to
\begin{small}
  \begin{align*}
\mathbb{E}[\nu_{t_n} \mid \mathcal{Y}_{t_n}] 
&= \Bigg\{ 
    (Q + 1)(1 - \kappa h)\, \mathbb{E}[\nu_{t_{n-1}} \mid \mathcal{Y}_{t_{n-1}}]  + (Q + 1)\left(\kappa \theta h + \rho\, \xi \left(Y_{t_n} - \mu h\right)\right)  \\
&\quad + \frac{1}{2h}(Y_{t_n} - \mu h)^2 
\Bigg\} \left(Q + \frac{3}{2}\right)^{-1},  \\
&=
    K\, \mathbb{E}[\nu_{t_{n-1}} \mid \mathcal{Y}_{t_{n-1}}]  + \Bigg(\frac{(Q + 1)}{B}(\kappa \theta h -\rho\, \xi\mu h)+\frac{\mu^2 h}{B}\Bigg) +\Bigg(\frac{(Q + 1)\rho\, \xi}{B}   - \mu \Bigg)Y_{t_n}\\
    &\quad  + \Bigg(\frac{1}{2hB}\Bigg)Y_{t_n}^2.
\end{align*}
\end{small}
The recursive back propagation of \(\mathbb{E}[\nu_{t_{n-1}} \mid \mathcal{Y}_{t_{n-1}}]\) introduces additional similar terms but at earlier discrete time points. Their aggregation leads to the expression in \eqref{eq:expected_nu}.
\section{Distribution of \(Q_{t_n|t_{n-1}}\) under the theoretical model}
\label{appendix:distributionofqt}
 The distribution of \( Q_{t_n|t_{n-1}} \) is modestly right-skewed. 
    \begin{figure}[H]
    \centering
\includegraphics[width=\linewidth]{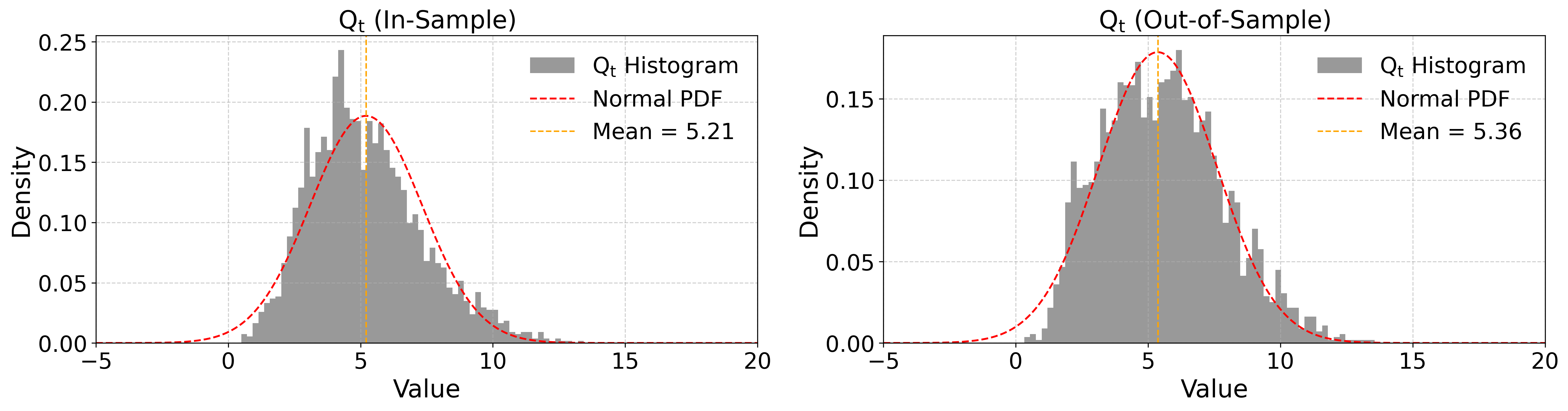}
\caption{Empirical distribution of \( Q_{t_n|t_{n-1}} \) under the theoretical model.}
        \label{fig:subfig-b}
    \end{figure}

\section{ADF algorithm on Heston volatility model}
\label{appendix:ADFalgorithm}
We illustrate the computational algorithm for the ADF in Algorithm \ref{alg:ADF1}.
\begin{algorithm}[H]
\caption{Iterative Assumed Density Filtering}
\begin{algorithmic}[1]
\Require Returns data $df$, step size $h$, parameters param with keys $\sigma_0$, $\kappa $, $\xi$, $\rho$, $\theta$
\State $\nu \gets \sigma_0^2$
\State $\kappa  \gets \text{param}[\kappa ]$, $\xi \gets \text{param}[\xi]$, $\mu \gets \text{mean}(df)$, $\rho \gets \text{param}[\rho]$, $\theta \gets \text{param}[\theta]$
\State Initialize: $\mathbb{E}_{\text{corr}} \gets [\nu]$, $V_{\text{corr}} \gets [\xi^2 h]$, $Q \gets [0]$
\For{each $t$ in index of $df$}
    \State \textbf{--- prediction step ---}
    \State $\mathbb{E}_{\text{pred}} \gets (\kappa \theta h) + (\rho\, \xi (df[t] - \mu h)) + (1-\kappa h) \mathbb{E}_{\text{corr}}[-1]$
    \State $V_{\text{pred}} \gets (1-\kappa h)^2 V_{\text{corr}}[-1] + \xi^2 (1-\rho^2) h \mathbb{E}_{\text{corr}}[-1]$
    \State $Q_{\text{pred}} \gets \frac{\mathbb{E}_{\text{pred}}^2}{V_{\text{pred}}}$
    \State $\alpha_{\text{pred}} \gets Q_{\text{pred}} + 2$
    \State $\beta_{\text{pred}} \gets (Q_{\text{pred}} + 1) \mathbb{E}_{\text{pred}}$
    \State \textbf{--- correction step ---}
    \State $\alpha_{\text{corr}} \gets \alpha_{\text{pred}} + 0.5$
    \State $\beta_{\text{corr}} \gets \beta_{\text{pred}} + \frac{(df[t] - \mu h)^2}{2h}$
    \State Append $ \frac{\beta_{\text{corr}}}{\alpha_{\text{corr}} - 1}$ to $\mathbb{E}_{\text{corr}}$
    \State Append $\frac{\mathbb{E}_{\text{corr}}[-1]^2}{\alpha_{\text{corr}} - 2}$ to $V_{\text{corr}}$
    \State Append $Q_{\text{pred}}$ to $Q$
\EndFor
\State \Return $\mathbb{E}_{\text{corr}}[1:],\; V_{\text{corr}}[1:],\; Q[1:]$
\end{algorithmic}
\label{alg:ADF1}
\end{algorithm}

The algorithm is lightweight because it avoids storing weighted sums of historical returns and instead updates quantities stepwise.

\section{Parametric bootstrap}
\label{appendix:D}

We use the parametric bootstrap to approximate standard errors: fix parameters at their full-sample estimates, simulate synthetic paths, re-estimate the model on each path, and repeat. This produces an empirical sampling distribution for each of the parameters. A standard bootstrap resamples time series over shorter windows from the in-sample time series. In path-dependent settings, this is unreliable because short windows lose some correlation information.

To assess the accuracy of the parametric bootstrap, we run two experiments with the simulated Heston model. First, we generate many paths from known parameters and estimate the model on each to obtain the true sampling distribution. Second, we estimate the model once on a single path, use those estimates to simulate new paths, and re-estimate to obtain an empirical distribution. 

Although the means of the two distributions may differ, their variances align closely (see Table \ref{table:bootstrap_se} and Figure~\ref{fig:parametricbootstrap}  showing the results of standard error estimates based on our Heston volatility model from Section 4.1).

\begin{table}[h]
\centering
\begin{tabular}{|c|c|c|c|c|}
\hline
Parameter & \(\kappa\) & \(\xi\) & \(\rho\) & \(\log \theta\) \\
\hline
Standard Bootstrap & 0.006238 & 0.000264 & 0.117105 & 0.438175 \\
\hline
Parametric Bootstrap & 0.008788 & 0.000491 & 0.112212 & 0.504947 \\
\hline
\end{tabular}
\caption{Approximate standard errors from standard and parametric bootstrap methods.}
\label{table:bootstrap_se}
\end{table}

Figure~\ref{fig:parametricbootstrap} compares the distributions of parameter estimates obtained from two approaches:
(i) standard bootstrap, where the process is re-simulated from the true data-generating Heston model with parameters described in Section 4.1,
and (ii) parametric bootstrap, where new synthetic paths simulated from the first estimate of the model are used to construct an empirical sampling distribution for the parameter estimates.
\begin{figure}[H]
    \centering
    \includegraphics[width=0.8\linewidth]{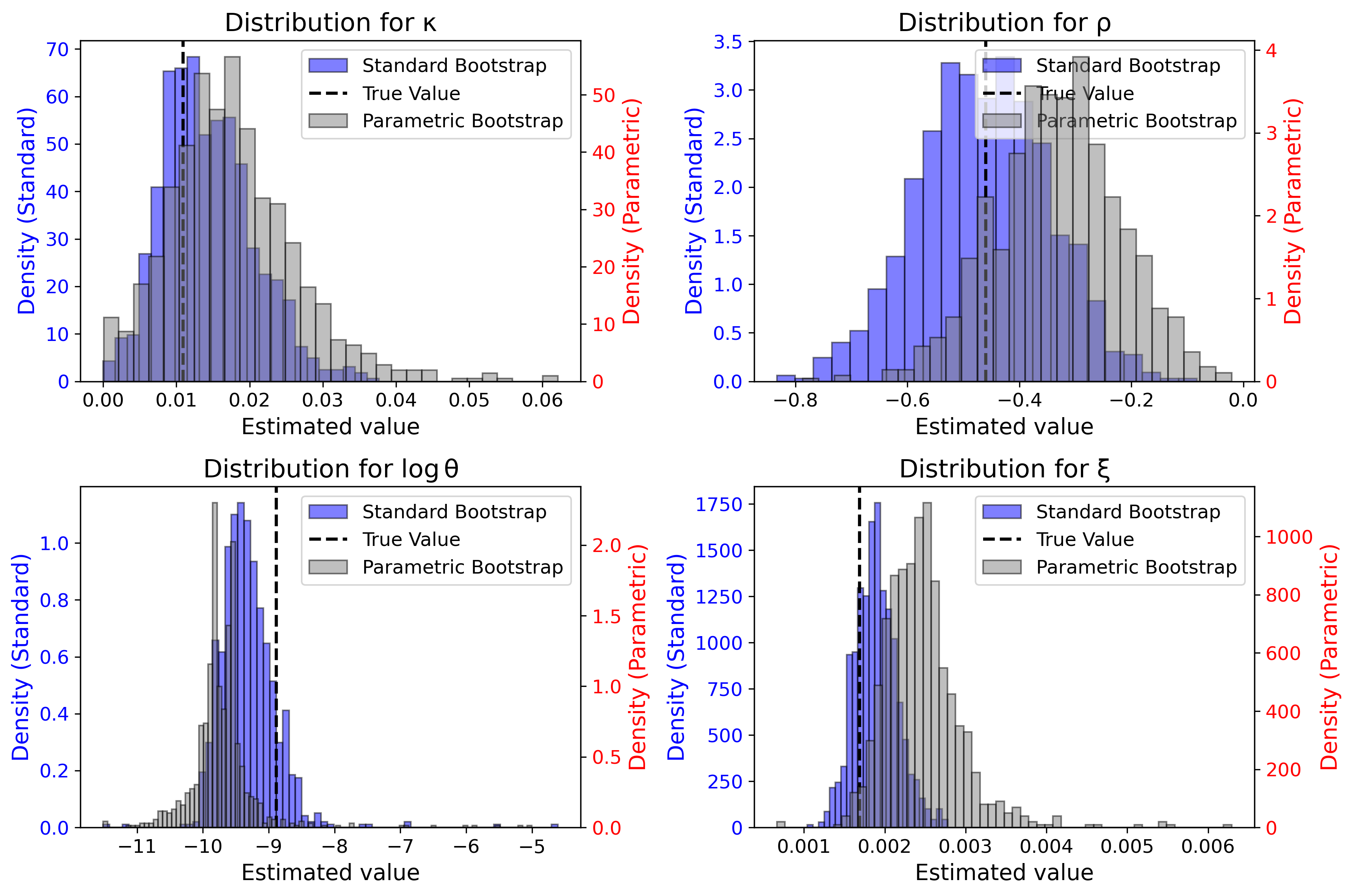}
    \caption{Distributions of parameter estimates obtained from the standard bootstrap (re-simulating from the true data-generating process) and the parametric bootstrap (re-simulating from the fitted model).}
\label{fig:parametricbootstrap}
\end{figure}

\end{appendix}
We find that the parameter variances are similar in each case. The true parameter values lie within the 95\% confidence intervals of the true sampling distribution. Although the parametric bootstrap shifts the center of the distribution, it provides a good approximation of the parameter variances.

\end{document}